\documentclass[letter,11pt]{article}

\ifdefined\DEBUG

\newcommand{\ari}[1]{\textcolor{blue}{#1}}
\newcommand{\wal}[1]{\textcolor{magenta}{#1}}
\def\rem#1{{\marginpar{\raggedright\scriptsize #1}}}
\newcommand{\fabr}[1]{\rem{\textcolor{red}{$\bullet$ #1}}}
\newcommand{\salr}[1]{\rem{\textcolor{green}{$\bullet$ #1}}}
\newcommand{\arir}[1]{\rem{\textcolor{blue}{$\bullet$ #1}}}
\newcommand{\walr}[1]{\rem{\textcolor{magenta}{$\bullet$ #1}}}
\else

\newcommand{\ari}[1]{#1}
\newcommand{\wal}[1]{#1}
\newcommand{\fabr}[1]{}
\newcommand{\salr}[1]{}
\newcommand{\arir}[1]{}
\newcommand{\walr}[1]{}
\fi

\usepackage[left=1in,right=1in,top=1in,bottom=1in]{geometry}
\usepackage[english]{babel}
\usepackage[utf8]{inputenc}
\usepackage{amsmath}
\usepackage{amsthm}
\usepackage{amsfonts}
\usepackage{graphicx}
\usepackage[colorinlistoftodos]{todonotes}
\usepackage{verbatim} 
\usepackage{enumitem}

\usepackage{fetamont}
\usepackage[T1]{fontenc}

\usepackage{times}
\usepackage{fullpage}

\usepackage{tikz,pgfplots}
\usetikzlibrary{patterns}
\usepackage{caption, setspace, subcaption}
\usepackage{authblk}

\usepackage{tikz,pgfplots}
\usetikzlibrary{patterns}
\newcommand{\eps}{\varepsilon}
\newcommand{\R}{\mathcal{R}}
\usepackage{lipsum}
\usepackage{amsfonts}
\usepackage{graphicx}
\usepackage{epstopdf}
\usepackage{algorithmic}
\usepackage{caption}
\usepackage{subcaption}
\usepackage{authblk}

\ifpdf
\DeclareGraphicsExtensions{.eps,.pdf,.png,.jpg}
\else
\DeclareGraphicsExtensions{.eps}
\fi

\newtheorem{theorem}{Theorem}
\newtheorem{lemma}[theorem]{Lemma}

\title{Improved Pseudo-Polynomial-Time Approximation for Strip Packing\footnote{A preliminary version of this paper appeared in the
Proceedings of the 36th IARCS Annual Conference on Foundations of Software Technology and
Theoretical Computer Science (FSTTCS 2016).}
\thanks{{\fontsize{11}{20}The authors from IDSIA are partially supported by ERC Starting Grant NEWNET 279352 and SNSF Grant APXNET 200021$\_$159697$/$1. 
Arindam Khan is  supported in part by the European Research Council, Grant Agreement No. 691672, the work was primarily done when the author was at IDSIA.}}}
\author[*]{Waldo G\'alvez}
\author[*]{Fabrizio Grandoni}
\author[*]{Salvatore Ingala}
\author[**]{Arindam Khan}
\affil[*]{IDSIA, USI-SUPSI, Switzerland\\
	\texttt{[waldo,fabrizio,salvatore]@idsia.ch}}
\affil[**]{Department of Computer Science, Technical University of Munich, Garching, Germany,
	\texttt{arindam.khan@in.tum.de}}

\date{\empty}

\begin{document}
\maketitle

\begin{abstract}
We study the {\it strip packing} problem, a classical packing problem which generalizes both {\it bin packing} and {\it makespan minimization}. Here we are given a set of axis-parallel rectangles in the two-dimensional plane and the goal is to pack them in a vertical strip of fixed width such that the height of the obtained packing is minimized. The packing must be {\it non-overlapping} and the rectangles cannot be {\it rotated}.

A reduction from the \emph{partition} problem shows that no approximation better than 3/2 is possible for strip packing in polynomial time (assuming P$\neq$NP). Nadiradze and Wiese [SODA16] overcame this barrier by presenting a $(\frac{7}{5}+\epsilon)$-approximation algorithm in pseudo-polynomial-time (PPT). As the problem is strongly NP-hard, it does not admit an exact PPT algorithm. \walr{should we add the approximation hardness in the abstract?\ari{I did not add as much happened after our publication. Fabrizio can decide. }}

In this paper we make further progress on the PPT approximability of strip packing, by presenting a $(\frac43+\epsilon)$-approximation algorithm. Our result is based on a non-trivial repacking of some rectangles in the \emph{empty space} left by the construction by Nadiradze and Wiese, and in some sense pushes their approach to its limit. 

Our PPT algorithm can be adapted to the case where we are allowed to rotate the rectangles by $90^\circ$, achieving the same approximation factor and breaking the polynomial-time approximation barrier of 3/2 for the case with rotations as well.
\end{abstract}



\section{Introduction}
\label{sec:intro}
In this paper, we consider the {\it strip packing} problem, a  well-studied classical two-dimensional packing problem \cite{BakerCR80,CGJT80,KenyonR00}. 
Here we are given a collection of rectangles, and an infinite vertical strip of width $W$ in the two dimensional  (2-D) plane. We need to find an axis-parallel embedding of the rectangles \emph{without rotations} inside the strip so that no two rectangles overlap (feasible \emph{packing}). Our goal is to minimize the total height of this packing. 

More formally, we are given a parameter $W\in \mathbb{N}$ and a set $\R=\{R_1,\ldots,R_n\}$ of  rectangles, each one characterized by a width $w_i\in \mathbb{N}$, $w_i\leq W$, and a height $h_i\in \mathbb{N}$. A packing of $\R$ is a pair $(x_i,y_i)\in \mathbb{N}\times \mathbb{N}$ for each $R_i$, with $0\leq x_i\leq W-w_i$, meaning that the left-bottom corner of $R_i$ is placed in position $(x_i,y_i)$ and its right-top corner in position $(x_i+w_i,y_i+h_i)$. This packing is feasible if the interiors of the rectangles are pairwise disjoint in this embedding (or equivalently rectangles are allowed to overlap on their boundary only). Our goal is to find a feasible packing of minimum \emph{height} $\max_i\{y_i+h_i\}$. 

Strip packing is a natural generalization of {\it one-dimensional bin packing} \cite{coffman2013bin} (when all the rectangles have the same height) and {\it makespan minimization} \cite{coffman1976computer} (when all the rectangles have the same width).
The problem has lots of applications in industrial engineering and computer science, specially in cutting stock, logistics and scheduling \cite{KenyonR00,harren20145}.
Recently, there have been a lot of applications of strip packing in electricity allocation and peak demand reduction in smart-grids \cite{tang2013smoothing, karbasioun2013power, ranjan2015offline}.

A simple reduction from the \emph{partition} problem shows that the problem cannot be approximated within a factor $\frac{3}{2}-\eps$ for any $\eps>0$ in polynomial-time unless P=NP.
This reduction relies on exponentially large (in $n$) rectangle widths.


Let $OPT=OPT(\R)$ denote the optimal height for the considered strip packing instance $(\R, W)$, and $h_{\max}=h_{\max}(\R)$ (resp. $w_{\max}=w_{\max}(\R)$) be the largest height (resp. width) of any rectangle in $\R$. Observe that trivially $OPT\geq h_{\max}$. W.l.o.g. we can assume that $W \leq n w_{max}$. The first non-trivial approximation algorithm for strip packing, with approximation ratio 3, was given by Baker, Coffman and Rivest \cite{BakerCR80}. 
The First-Fit-Decreasing-Height algorithm (FFDH) by Coffman et al.~\cite{CGJT80} gives a 2.7 approximation. Sleator \cite{sleator19802} gave an algorithm that generates packing of height $2OPT+\frac{h_{max}}{2}$, hence achieving a 2.5 approximation. 
Afterwards, Steinberg \cite{steinberg1997strip} and Schiermeyer \cite{schiermeyer1994reverse} independently improved the approximation ratio to 2. Harren and van Stee \cite{harren2009improved} first broke the barrier of 2 with their 1.9396 approximation.
The present best $(\frac53+\eps)$-approximation is due to Harren et al. \cite{harren20145}.

Recently algorithms running in pseudo-polynomial time (PPT) for this problem have been developed. More specifically, the running time of a PPT algorithm for Strip Packing is $O((Nn)^{O(1)})$, where $N=\max\{w_{max},h_{max}\}$\footnote{For the case without rotations, the polynomial dependence on $h_{max}$ can indeed be removed with standard techniques.}. First, Jansen and Th\"ole~\cite{JT10} showed a PPT $(3/2+\eps)$-approximation algorithm, and later Nadiradze and Wiese \cite{NW16} overcame the $\frac{3}{2}$-inapproximability barrier by presenting a PPT $(\frac{7}{5}+\epsilon)$-approximation algorithm. 
As strip packing is strongly NP-hard \cite{garey1978strong}, it does not admit an exact PPT algorithm. 
\walr{please check this paragraph. Also I'm not sure how to include the 4/3-apx from Jansen and Rau, maybe should go in "our contribution". \ari{I added the recent progress at the end of Section 1.2 (Related work)}}

\subsection{Our contribution and techniques} 
In this paper, we make progress on the PPT approximability of strip packing, by presenting an improved $(\frac43+\eps)$ approximation. Our approach refines the technique of Nadiradze and Wiese \cite{NW16}, that modulo several technical details works as follows: let $\alpha\in [1/3,1/2)$ be a proper constant parameter, and define a rectangle $R_i$ to be \emph{tall} if $h_i> \alpha\cdot OPT$. They prove that the optimal packing can be structured into a constant number of axis-aligned rectangular regions (\emph{boxes}), that occupy a total height of $OPT' \leq (1+\eps) OPT$ inside the vertical strip. Some rectangles are not fully contained into one box (they are \emph{cut} by some box). Among them, tall rectangles remain in their original position. All the other cut rectangles are repacked on top of the boxes: part of them in a horizontal box of size $W\times O(\eps)OPT$, and the remaining ones in a vertical box of size $O(\eps W)\times \alpha\,OPT$ (that we next imagine as placed on the top-left of the packing under construction).

Some of these boxes contain only relatively high rectangles (including tall ones) of relatively small width.
The next step is a rearrangement of the rectangles inside one such \emph{vertical} box $\overline{B}$ (see Figure \ref{fig_pseudo-rectangles1}), say of size $\overline{w} \times \overline{h}$: they first slice non-tall rectangles into unit width rectangles
(this slicing can be finally avoided with standard techniques). Then they shift tall rectangles to the top/bottom of $\overline{B}$, shifting sliced rectangles consequently (see Figure \ref{fig_pseudo-rectangles2}). Now they discard all the (sliced) rectangles completely contained in a central horizontal region of size $\overline{w}\times (1+\eps-2\alpha)\overline{h}$, and they \emph{nicely rearrange} the remaining rectangles into a constant number of \emph{sub-boxes} (excluding possibly a few more non-tall rectangles, that can be placed in the additional vertical box).

These discarded rectangles can be packed into $2$ extra boxes of size $\frac{\overline{w}}{2}\times(1+\eps-2\alpha)\overline{h}$ (see Figure 
\ref{fig_pseudo-rectangles4}).
In turn, the latter boxes can be packed into two \emph{discarded} boxes of size $\frac{W}{2}\times (1+\eps-2\alpha)OPT'$, that we can imagine as placed, one on top of the other, on the top-right of the packing. See Figure \ref{fig_packing_NW} for an illustration of the final packing. This leads to a total height of $(1+\max\{\alpha,2(1-2\alpha)\}+O(\eps))\cdot OPT$, which is minimized by choosing $\alpha=\frac{2}{5}$. 

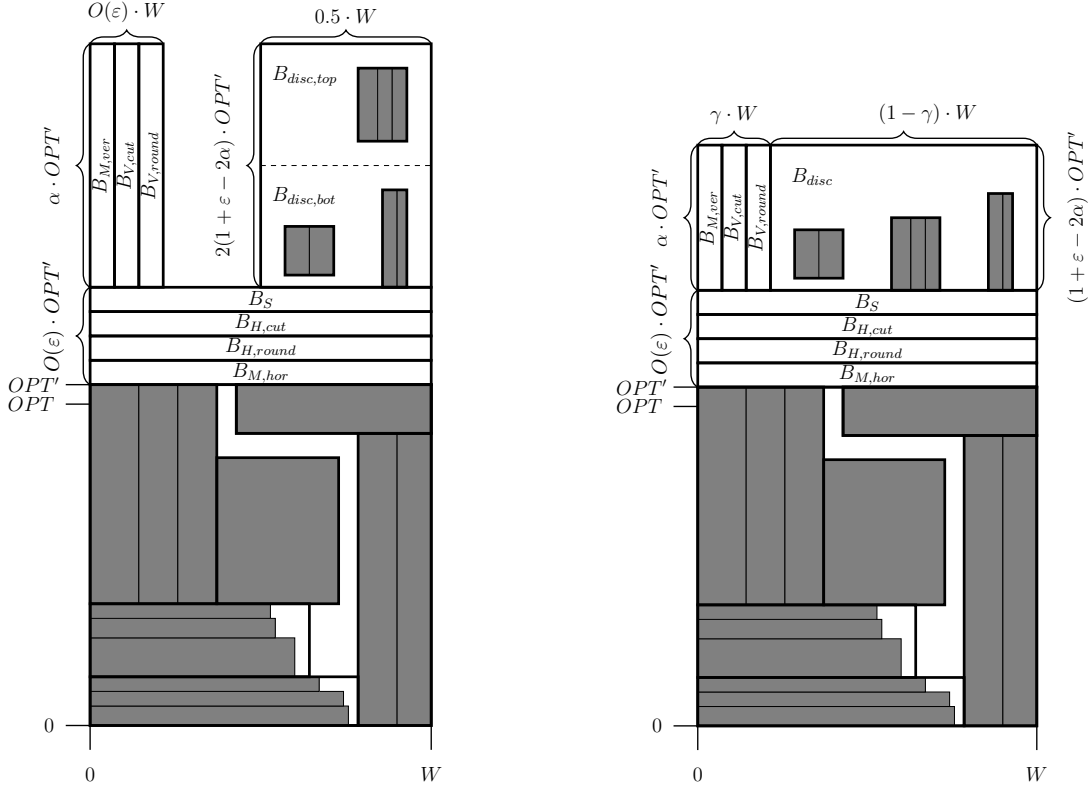
\begin{figure}
	\captionsetup[subfigure]{justification=centering}
	\hspace{-13pt}
	\begin{subfigure}[b]{.53\textwidth}
	\resizebox{6cm}{!}{
\begin{tikzpicture}


\draw[fill = gray] (0,0) rectangle (5.3,0.4);
\draw[fill = gray] (0,0.4) rectangle (5.2,0.7);
\draw[fill = gray] (0,0.7) rectangle (4.7,1);

\draw[fill = gray] (0,1) rectangle (4.2,1.8);
\draw[fill = gray] (0,1.8) rectangle (3.8,2.2);
\draw[fill = gray] (0,2.2) rectangle (3.7,2.5);

\draw[fill = gray] (0,2.5) rectangle (1,7);
\draw[fill = gray] (1,2.5) rectangle (1.8,7);
\draw[fill = gray] (1.8,2.5) rectangle (2.6,7);

\fill[fill = gray] (2.6,2.5) rectangle (5.1,5.5);

\draw[fill=gray] (5.5,0) rectangle (6.3,6);
\draw[fill=gray] (6.3,0) rectangle (7,6);

\fill[fill=gray] (3,6) rectangle (7,7);

\draw[fill=gray] (4,9.25) rectangle (4.5,10.25);
\draw[fill=gray] (4.5,9.25) rectangle (5,10.25);

\draw[fill=gray] (6,9) rectangle (6.3,11);
\draw[fill=gray] (6.3,9) rectangle (6.5,11);

\draw[fill=gray] (5.5,12) rectangle (5.9,13.5);
\draw[fill=gray] (5.9,12) rectangle (6.2,13.5);
\draw[fill=gray] (6.2,12) rectangle (6.5,13.5);


\draw[ultra thick] (0,0) rectangle (7,7);

\draw[ultra thick] (0,0) rectangle (5.5,1);
\draw[ultra thick] (0,1) rectangle (4.5,2.5);
\draw[ultra thick] (0,2.5) rectangle (2.6,7);
\draw[ultra thick] (2.6,2.5) rectangle (5.1,5.5);
\draw[ultra thick] (3,6) rectangle (7,7);
\draw[ultra thick] (5.5,0) rectangle (7,6);
\draw[ultra thick] (3,6) rectangle (7,7);

\draw[ultra thick] (0,7) rectangle (7,7.5);
\draw[ultra thick] (0,7.5) rectangle (7,8);
\draw[ultra thick] (0,8) rectangle (7,8.5);
\draw[ultra thick] (0,8.5) rectangle (7,9);

\draw[ultra thick] (0,9) rectangle (0.5,14);
\draw[ultra thick] (0.5,9) rectangle (1,14);
\draw[ultra thick] (1,9) rectangle (1.5,14);

\draw[ultra thick] (3.5,9) rectangle (7,14);

\draw[ultra thick] (4,9.25) rectangle (5,10.25);
\draw[ultra thick] (6,9) rectangle (6.5,11);

\draw[ultra thick] (5.5,12) rectangle (6.5,13.5);


\draw (3.6,13.35) node[anchor=west] {\large \textbf{$B_{disc, top}$}};
\draw (3.6,10.85) node[anchor=west] {\large \textbf{$B_{disc, bot}$}};
\draw (0.25,11.5) node [rotate=90] {\large \textbf{$B_{M, ver}$}};
\draw (0.75,11.5) node [rotate=90] {\large \textbf{$B_{V, cut}$}};
\draw (1.25,11.5) node [rotate=90] {\large \textbf{$B_{V, round}$}};
\draw (3.5,8.75) node {\large \textbf{$B_S$}};
\draw (3.5,8.25) node {\large \textbf{$B_{H, cut}$}};
\draw (3.5,7.75) node {\large \textbf{$B_{H, round}$}};
\draw (3.5,7.25) node {\large \textbf{$B_{M, hor}$}};

\draw[dashed] (3.5,11.5) -- (7,11.5);

\draw[thick] (0,0) -- (-0.5,0);
\draw (-0.6,0) node[anchor = east] {\large $0$}; 

\draw[thick] (0,0) -- (0,-0.5);
\draw (0,-1) node {\large $0$}; 

\draw[thick] (7,0) -- (7,-0.5);
\draw (7,-1) node {\large $W$}; 

\draw[thick] (0,6.6) -- (-0.5,6.6);
\draw (-0.6,6.6) node[anchor = east] {\large $OPT$}; 

\draw[thick] (0,7) -- (-0.5,7);
\draw (-0.5,7) node[anchor = east] {\large $OPT'$}; 


\draw [thick,decorate,decoration={brace,amplitude=8pt}] 
(0,7) -- (0,9); 
\draw (-0.4,8.4) node [rotate=90,anchor = south] {\resizebox{65pt}{!}{$O(\eps) \cdot OPT'$}}; 

\draw [thick,decorate,decoration={brace,amplitude=8pt}] 
(3.5,9) -- (3.5,14); 
\draw (3.1,11.5) node [anchor = south, rotate=90] {\large $2(1+\eps-2\alpha) \cdot OPT'$}; 

\draw [thick,decorate,decoration={brace,amplitude=8pt}] 
(0,14) -- (1.5,14); 
\draw (0.75,14.3) node [anchor = south] {\large $O(\eps) \cdot W$}; 

\draw [thick,decorate,decoration={brace,amplitude=8pt}] 
(3.5,14) -- (7,14); 
\draw (5.25,14.3) node [anchor = south] {\large $0.5 \cdot W$};

\draw [thick,decorate,decoration={brace,amplitude=8pt}] 
(0,9) -- (0,14); 
\draw (-0.5,11.5) node [rotate=90,anchor = south] {\large $\alpha \cdot OPT'$}; 

\end{tikzpicture}}
        \caption{Final packing obtained by \\ Nadiradze \& Wiese \cite{NW16}.}
        \label{fig_packing_NW}
	\end{subfigure}%
	\hspace{-21pt}
	\begin{subfigure}[b]{.55\textwidth}
\resizebox{6.5cm}{!}{
\begin{tikzpicture}


\draw[fill = gray] (0,0) rectangle (5.3,0.4);
\draw[fill = gray] (0,0.4) rectangle (5.2,0.7);
\draw[fill = gray] (0,0.7) rectangle (4.7,1);

\draw[fill = gray] (0,1) rectangle (4.2,1.8);
\draw[fill = gray] (0,1.8) rectangle (3.8,2.2);
\draw[fill = gray] (0,2.2) rectangle (3.7,2.5);

\draw[fill = gray] (0,2.5) rectangle (1,7);
\draw[fill = gray] (1,2.5) rectangle (1.8,7);
\draw[fill = gray] (1.8,2.5) rectangle (2.6,7);

\fill[fill = gray] (2.6,2.5) rectangle (5.1,5.5);

\draw[fill=gray] (5.5,0) rectangle (6.3,6);
\draw[fill=gray] (6.3,0) rectangle (7,6);

\fill[fill=gray] (3,6) rectangle (7,7);

\draw[fill=gray] (2,9.25) rectangle (2.5,10.25);
\draw[fill=gray] (2.5,9.25) rectangle (3,10.25);

\draw[fill=gray] (6,9) rectangle (6.3,11);
\draw[fill=gray] (6.3,9) rectangle (6.5,11);

\draw[fill=gray] (4,9) rectangle (4.4,10.5);
\draw[fill=gray] (4.4,9) rectangle (4.7,10.5);
\draw[fill=gray] (4.7,9) rectangle (5,10.5);


\draw[ultra thick] (0,0) rectangle (7,7);

\draw[ultra thick] (0,0) rectangle (5.5,1);
\draw[ultra thick] (0,1) rectangle (4.5,2.5);
\draw[ultra thick] (0,2.5) rectangle (2.6,7);
\draw[ultra thick] (2.6,2.5) rectangle (5.1,5.5);
\draw[ultra thick] (3,6) rectangle (7,7);
\draw[ultra thick] (5.5,0) rectangle (7,6);
\draw[ultra thick] (3,6) rectangle (7,7);

\draw[ultra thick] (0,7) rectangle (7,7.5);
\draw[ultra thick] (0,7.5) rectangle (7,8);
\draw[ultra thick] (0,8) rectangle (7,8.5);
\draw[ultra thick] (0,8.5) rectangle (7,9);

\draw[ultra thick] (0,9) rectangle (0.5,12);
\draw[ultra thick] (0.5,9) rectangle (1,12);
\draw[ultra thick] (1,9) rectangle (1.5,12);

\draw[ultra thick] (1.5,9) rectangle (7,12);

\draw[ultra thick] (2,9.25) rectangle (3,10.25);
\draw[ultra thick] (4,9) rectangle (5,10.5);
\draw[ultra thick] (6,9) rectangle (6.5,11);

\draw (1.8,11.35) node[anchor=west] {\large \textbf{$B_{disc}$}};
\draw (0.25,10.5) node [rotate=90] {\large \textbf{$B_{M, ver}$}};
\draw (0.75,10.5) node [rotate=90] {\large \textbf{$B_{V, cut}$}};
\draw (1.25,10.5) node [rotate=90] {\large \textbf{$B_{V, round}$}};
\draw (3.5,8.75) node {\large \textbf{$B_S$}};
\draw (3.5,8.25) node {\large \textbf{$B_{H, cut}$}};
\draw (3.5,7.75) node {\large \textbf{$B_{H, round}$}};
\draw (3.5,7.25) node {\large \textbf{$B_{M, hor}$}};

\draw[thick] (0,0) -- (-0.5,0);
\draw (-0.6,0) node[anchor = east] {\large $0$}; 

\draw[thick] (0,0) -- (0,-0.5);
\draw (0,-1) node {\large $0$}; 

\draw[thick] (7,0) -- (7,-0.5);
\draw (7,-1) node {\large $W$}; 

\draw[thick] (0,6.6) -- (-0.5,6.6);
\draw (-0.6,6.6) node[anchor = east] {\large $OPT$}; 

\draw[thick] (0,7) -- (-0.5,7);
\draw (-0.5,7) node[anchor = east] {\large $OPT'$}; 


\draw [thick,decorate,decoration={brace,amplitude=8pt}] 
(0,7) -- (0,9); 
\draw (-0.4,8.4) node [rotate=90,anchor = south] {\resizebox{65pt}{!}{$O(\eps) \cdot OPT'$}}; 

\draw [thick,decorate,decoration={brace,amplitude=8pt}] 
(0,9) -- (0,12); 
\draw (-0.5,10.7) node [rotate=90, anchor = south] {\large $\alpha \cdot OPT'$}; 

\draw [thick,decorate,decoration={brace,amplitude=8pt}] 
(0,12) -- (1.5,12); 
\draw (0.75,12.3) node [anchor = south] {\large $\gamma \cdot W$}; 

\draw [thick,decorate,decoration={brace,amplitude=8pt}] 
(1.5,12) -- (7,12); 
\draw (4.75,12.3) node [anchor = south] {\large $(1-\gamma) \cdot W$};

\draw [thick,decorate,decoration={brace,amplitude=8pt}] 
(7,12) -- (7,9); 
\draw (7.5,10.5) node [rotate=90,anchor = north] {\large $(1+\eps - 2\alpha) \cdot OPT'$}; 

\end{tikzpicture}}
        \caption{Final packing obtained in this work. \\Here $\gamma$ is a small constant depending on $\eps$.}
        \label{fig_our_packing}
	\end{subfigure}
	\caption{Comparison of final solutions.}
	\label{fig_structure}
\end{figure}

Our main technical contribution is a repacking lemma that allows one to repack a small fraction of the discarded rectangles of a given box inside the free space left by the corresponding sub-boxes (while still having $O_{\eps}(1)$ many sub-boxes in total). This is illustrated in Figure \ref{fig_pseudo-rectangles5}. This way we can pack all the discarded rectangles into a \emph{single} discarded box of size $(1-\gamma)W\times (1+\eps-2\alpha)OPT'$, where $\gamma$ is a small constant depending on $\eps$, that we can place on the top-right of the packing. The vertical box where the remaining rectangles are packed still fits to the top-left of the packing, next to the discarded box. See Figure \ref{fig_our_packing} for an illustration. Choosing $\alpha=1/3$ gives the claimed approximation factor. 

We remark that the basic approach by Nadiradze and Wiese strictly requires that at most $2$ tall rectangles can be packed one on top of the other in the optimal packing, hence imposing $\alpha\geq 1/3$. Thus in some sense this work pushes their approach to its limit.

The algorithm by Nadiradze and Wiese \cite{NW16} is not directly applicable to the case when $90^\circ$ rotations are allowed. 
In particular, they use a linear program to pack some rectangles. When rotations are allowed, it is unclear how to decide which rectangles are packed by the linear program.
We use a combinatorial \emph{container}-based approach to circumvent this limitation, which allows us to pack all the rectangles using dynamic programming. This way we achieve a PPT $(4/3+\eps)$-approximation for strip packing with rotations, breaking the polynomial-time approximation barrier of 3/2 for that variant as well.

\subsection{Related work}

For packing problems, many pathological lower bound instances occur when $OPT$ is small. Thus it is often insightful to consider the \textit{asymptotic approximation ratio}. Coffman et al.~\cite{CGJT80} 
described two {\it level-oriented} algorithms, Next-Fit-Decreasing-Height (NFDH) and First-Fit-Decreasing-Height (FFDH), that
achieve asymptotic approximations of 2 and 1.7, respectively. After a sequence of improvements
\cite{golan1981performance, baker198154}, the seminal work of Kenyon and R{\'e}mila \cite{KenyonR00} provided an asymptotic polynomial-time approximation scheme (APTAS) with an additive term $O\left(\frac{h_{max}}{\eps^2}\right)$. The latter additive term was subsequently improved to $h_{max}$ by Jansen and Solis-Oba \cite{jansen2009rectangle}.

In the variant of strip packing \emph{with rotations}, we are allowed to rotate the input rectangles by $90^\circ$ (in other terms, we are free to swap the width and height of an input rectangle). The case with rotations is much less studied in the literature. It seems that most of the techniques that work for the case without rotations can be extended to the case with rotations, however this is not always a trivial task. In particular, it is not hard to achieve a $2+\eps$ approximation, and the $3/2$ hardness of approximation extends to this case as well~\cite{jansen2009rectangle}. In terms of asymptotic approximation, Miyazawa and Wakabayashi \cite{miyazawa2004packing} gave an algorithm with
asymptotic performance ratio of 1.613.
Later, Epstein and van Stee \cite{epstein2006side} gave a $\frac32$ asymptotic approximation.
Finally, Jansen and van Stee \cite{jansen2005strip} achieved an APTAS for the case with rotations.

Strip packing has also been well studied for higher dimensions.
The present best asymptotic approximation for 3-D strip packing is due to Jansen and Pr{\"{a}}del \cite{JansenP14} who presented a 1.5-approximation extending techniques from 2-D bin packing.

There are many other related geometric packing problems. For example, in the \emph{independent set of rectangles} problem we are given a collection of axis-parallel rectangles embedded in the plane, and we need to find a maximum cardinality/weight subset of non-overlapping rectangles \cite{AW13,CC09,CH12}. Interesting connections between this problem and the \emph{unsplittable flow on a path} problem were recently discovered \cite{AGLW14,BCES06,BSW14}. In the \emph{geometric knapsack} problem we wish to pack a maximum cardinality/profit subset of the rectangles in a given square knapsack \cite{AdamaszekW15,GGHIKW17,JansenZ07}. \arir{added our FOCS paper.}One can also consider a natural geometric version of bin packing, where the goal is to pack a given set of rectangles in the smallest possible number of square bins \cite{BansalK14}. We refer the readers to \cite{ChristensenKPT17, KhanThesis} for surveys on geometric packing problems. 

\textbf{Subsequent Progress:}
\ari{Since the publication of our extended abstract \cite{GalvezGIK16},  new results have appeared.
Adamaszek et al.~\cite{AKPP17} proved that there is no PPT $(\frac{12}{11}-\eps)$-approximation algorithm for Strip Packing unless $NP\subseteq DTIME(2^{\textnormal{polylog}(n)})$. On the other hand,
Jansen and Rau \cite{JansenR17} independently showed a PPT $(4/3+\eps)$-approximation algorithm with running time $(nW)^{1/\eps^{O(2^{1/\eps})}}$ \wal{for the case without rotations}. 
}
Very recently, new results have been announced \cite{HJRS17, Raunew}  claiming to give a tight $(5/4+\eps)$-approximation algorithm. 
\arir{I did not mention the 5/4 result that Klaus told us as it is not yet public.\wal{Probably the 5/4 hardness shouldn't be mentioned either since it's only in arxiv}}

\subsection{Organization of the paper}
First, we discuss some preliminaries and notations in Section~\ref{sec:prelim}. 
Section~\ref{sec:repack} contains our main technical contribution, the \emph{repacking lemma}. Then, in Section~\ref{sec:structural_lemma}, we discuss a refined structural result leading to a packing into $O_\eps(1)$ many \emph{containers}. In Section~\ref{sec:algo}, we describe our algorithm to pack the rectangles and in Section~\ref{sec:rot} we extend our algorithm to the case with rotations. Finally, in Section~\ref{sec:conc}, we conclude with some observations.

\section{Preliminaries and notations}\label{sec:prelim}

Throughout the present work, we will follow the notation from \cite{NW16}, which will be explained as it is needed.

Recall that $OPT\in \mathbb{N}$ denotes the height of the optimal packing for instance $\R$. By trying all the pseudo-polynomially many possibilities, we can assume that $OPT$ is known to the algorithm. Given a set $\mathcal{M}\subseteq\R$ of rectangles, $a(\mathcal{M})$ will denote the total area of rectangles in $\mathcal{M}$, i.e., $a(\mathcal{M}) = \sum_{R_i\in \mathcal{M}}{h_i \cdot w_i}$, and $h_{\max}(\mathcal{M})$ (resp. $w_{\max}(\mathcal{M})$) denotes the maximum height (resp. width) of rectangles in $\mathcal{M}$. Throughout this work, a \emph{box} of size $a\times b$ means an axis-aligned rectangular region of width $a$ and height $b$. 

In order to lighten the notation, we sometimes interpret a rectangle/box as the corresponding region inside the strip  according to some given embedding. The latter embedding will not be specified when clear from the context. Similarly, we sometimes describe an embedding of some rectangles inside a box, and then embed the box inside the strip: the embedding of the considered rectangles is shifted consequently in that case.

A vertical (resp. horizontal) \emph{container} is an axis-aligned rectangular region where we implicitly assume that rectangles are packed one next to the other from left to right (resp., bottom to top), i.e., any vertical (resp. horizontal) line intersects only one packed rectangle (see Figure~\ref{fig_container}). Container-like packings will turn out to be particularly useful since they naturally induce a (one-dimensional) knapsack instance.

\subsection{Classification of rectangles}

Let $0 < \eps < \alpha$, and assume for simplicity that $\frac{1}{\eps} \in \mathbb{N}$. We first classify the input rectangles into six  groups according to parameters $\delta_h, \delta_w, \mu_h, \mu_w$ satisfying $\eps\geq \delta_h > \mu_h > 0$ and $\eps\geq \delta_w > \mu_w > 0$, whose values will be chosen later (see also Figure \ref{fig_classification}). A rectangle $R_i$ is 
\begin{figure}
	\captionsetup[subfigure]{justification=centering}
    \hspace{-21pt}
	\begin{subfigure}[b]{.38\textwidth}
		\centering
\resizebox{!}{5.45cm}{\begin{tikzpicture}

\fill[pattern color = lightgray, pattern = north east lines] (0,5.25) rectangle (3,8);
\fill[pattern color = lightgray, pattern = north west lines] (3,3.5) rectangle (6,8);
\fill[pattern color = lightgray, pattern = north west lines] (0,3.5) rectangle (1.5,5.25);
\fill[color = lightgray] (1.5,0) rectangle (3,1.75);
\fill[color = lightgray] (0,1.75) rectangle (6,3.5);
\fill[color = lightgray] (1.5,3.5) rectangle (3,5.25);
\fill[pattern color = lightgray, pattern = north east lines] (0,0) rectangle (1.5,1.75);
\fill[pattern color = lightgray, pattern = north east lines] (3,0) rectangle (6,1.75);


\draw[dashed] (1.5,0) -- (1.5,8.5);
\draw[dashed] (3,0) -- (3,8.5);
\draw[dashed] (6,0) -- (6,8.5);

\draw[dashed] (0,1.75) -- (6.5,1.75);
\draw[dashed] (0,3.5) -- (6.5,3.5);
\draw[dashed] (0,5.25) -- (6.5,5.25);
\draw[dashed] (0,8) -- (6.5,8);


\draw[->] (0,0) -- (6.5,0);
\draw[->] (0,0) -- (0,8.5);


\draw[ultra thick] (0,0) rectangle (1.5,1.75);
\draw[ultra thick] (1.5,0) -- (1.5,1.75) -- (0,1.75) -- (0,3.5) -- (1.5,3.5) -- (1.5,5.25) -- (3,5.25) -- (3,3.5) -- (6,3.5) -- (6,1.75) -- (3,1.75) -- (3,0) -- (1.5,0);
\draw[ultra thick] (3,0) rectangle (6,1.75);
\draw[ultra thick] (0,3.5) rectangle (1.5,5.25);
\draw[ultra thick] (0,5.25) rectangle (3,8);
\draw[ultra thick] (3,3.5) rectangle (6,8);


\draw (0,0) -- (0,-0.25);
\draw (0,-0.25) node[anchor=north] {$0$};
\draw (1.5,0) -- (1.5,-0.25);
\draw (1.5,-0.25) node[anchor=north] {$\mu_w W$};
\draw (3,0) -- (3,-0.25);
\draw (3,-0.25) node[anchor=north] {$\delta_w W$};
\draw (6,0) -- (6,-0.25);
\draw (6,-0.25) node[anchor=north] {$W$};

\draw (0,0) -- (-0.25,0);
\draw (-0.25,0) node[anchor=east] {$0$};
\draw (0,1.75) -- (-0.25,1.75);
\draw (-0.25,1.75) node[anchor=east] {$\mu_h OPT$};
\draw (0,3.5) -- (-0.25,3.5);
\draw (-0.25,3.5) node[anchor=east] {$\delta_h OPT$};
\draw (0,5.25) -- (-0.25,5.25);
\draw (-0.25,5.25) node[anchor=east] {$\alpha OPT$};
\draw (0,8) -- (-0.25,8);
\draw (-0.25,8) node[anchor=east] {$OPT$};


\draw (0.75,0.875) node {\textbf{\small Small}};
\draw (0.75,2.625) node {\textbf{\small Medium}};
\draw (0.75,4.375) node {\textbf{\small Vertical}};
\draw (0.75,6.625) node {\textbf{\small Tall}};

\draw (2.25,0.875) node {\textbf{\small Medium}};
\draw (2.25,2.625) node {\textbf{\small Medium}};
\draw (2.25,4.375) node {\textbf{\small Medium}};
\draw (2.25,6.625) node {\textbf{\small Tall}};

\draw (4.5,0.875) node {\textbf{\small Horizontal}};
\draw (4.5,2.625) node {\textbf{\small Medium}};
\draw (4.5,4.375) node {\textbf{\small Large}};
\draw (4.5,6.625) node {\textbf{\small Large}};

\end{tikzpicture}}
		\caption{Each rectangle is represented \\as a point on the plane with $x$ \\(resp., $y$) coordinate indicating \\its width (resp., height).}
		\label{fig_classification}
	\end{subfigure}
	\begin{subfigure}[b]{.29\textwidth}
        \centering
\resizebox{!}{4cm}{
\begin{tikzpicture}


\draw[solid, fill=gray] (0,0) rectangle (0.5,4.9);
\draw[solid, fill=gray] (0.5,0) rectangle (1.2,4.7);
\draw[solid, fill=gray] (1.2,0) rectangle (1.7,4.2);
\draw[solid, fill=gray] (1.7,0) rectangle (2.3,3.8);


\draw[solid] (0,0) rectangle (2.5,5);

\end{tikzpicture}}
		\caption{Example of \emph{vertical container}. Every vertical line intersects at most one rectangle.}
		\label{fig_container}
	\end{subfigure}%
    \hspace{2pt}
	\begin{subfigure}[b]{.35\textwidth}
		\centering
\resizebox{!}{4cm}{
\begin{tikzpicture}


\draw[solid, fill=gray] (-0.5,1.5) rectangle (0.75,4);
\draw[solid, fill=gray] (1.5,4.5) rectangle (4.5,5.5);


\draw[pattern=north east lines, pattern color=black] (5.5,3) rectangle (6.5,6.5);
\draw[pattern=north east lines, pattern color=black] (5,-0.75) rectangle (7.5,0.5);


\draw[solid, fill=lightgray] (2,1) rectangle (3.5,3.5);


\draw[solid] (0,0) rectangle (7,5);


\draw (3.5,5.75) node [anchor=south] {\LARGE Box $B$};

\end{tikzpicture}}
        \caption{Gray rectangles are \emph{nicely cut} by $B$, dashed ones are \emph{cut} but \emph{not nicely cut} by $B$, and light gray one is not cut by $B$.}
		\label{fig_cut_rectangles}
	\end{subfigure}
	\caption{Illustration of some of the definitions used in this work.}
	\label{fig_definitions}
\end{figure}
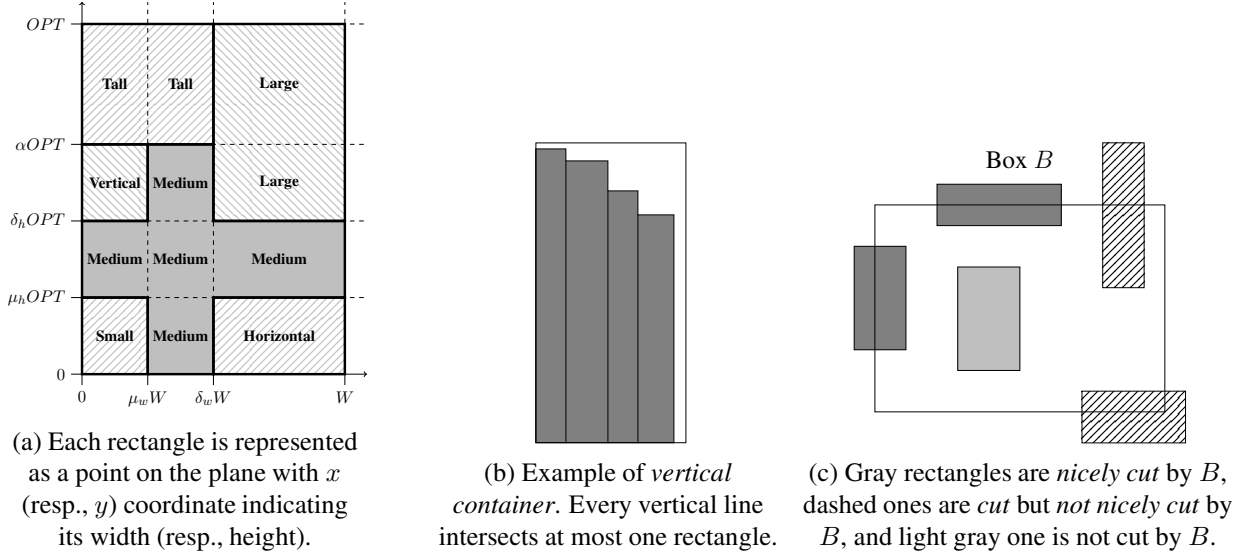

\begin{itemize}
	\item \emph{Large} if $h_i \ge \delta_h OPT$ and $w_i \ge \delta_w W$.
	\item \emph{Tall} if $h_i > \alpha OPT$ and $w_i < \delta_w W$.
	\item \emph{Vertical} if $h_i \in [\delta_h OPT, \alpha OPT]$ and $w_i \le \mu_w W$,
	\item \emph{Horizontal} if $h_i \le \mu_h OPT$ and $w_i \ge \delta_w W$,
		\item \emph{Small} if $h_i \le \mu_h OPT$ and $w_i \le \mu_w W$;
	\item \emph{Medium} in all the remaining cases, i.e., if $h_i \in (\mu_h OPT, \delta_h OPT)$, or $w_i \in (\mu_w W, \delta_w W)$ and $h_i \le \alpha OPT$.
\end{itemize}
We use $L$, $T$, $V$, $H$, $S$, and $M$ to denote large, tall, vertical, horizontal, small, and medium rectangles, respectively.
We remark that, differently from \cite{NW16}, we need to allow $\delta_h\neq \delta_w$ and $\mu_h\neq\mu_w$ due to some additional constraints in our construction (see Section~\ref{sec:algo}).

Notice that according to this classification, every vertical line across the optimal packing intersects at most two tall rectangles. The following lemma allows us to choose $\delta_h, \delta_w, \mu_h$ and $\mu_w$ in such a way that $\delta_h$ and $\mu_h$ ($\delta_w$ and $\mu_w$, respectively) differ by a large factor, and medium rectangles have small total area.

\begin{lemma}\label{lem:mediumrectanglesarea}
Given a polynomial-time computable function $f : (0, 1) \rightarrow (0, 1)$, with $f(x) < x$, any constant $\eps\in (0,1)$, and any positive integer $k$, we can compute in polynomial time a set $\Delta$ of $T=2(\frac{1}{\eps})^k$ many positive real numbers upper bounded by $\eps$, such that there is at least one number $\delta_h \in \Delta$ so that 
$a(M)\leq \eps^k \cdot OPT \cdot W$ by choosing
$\mu_h = f(\delta_h)$, $\mu_w=\frac{\eps \mu_h}{12}$, and $\delta_w=\frac{\eps \delta_h}{12}$.

\end{lemma}

\begin{proof} Let $T = 2(\frac{1}{\eps})^k$. Let $y_1 = \eps$, and, for each $j \in \{1,\dots,T\}$, define $y_{j+1} = f(y_j)$. Let $x_j = \frac{\eps y_j}{12}$. For each $j \leq T$, let $W_j = \{R_i \in \R : w_i \in [x_{i+1}, x_i)\}$ and similarly $H_j = \{R_i \in \R : h_i \in [y_{i+1}, y_i)\}$. Observe that $W_{j'}$ is disjoint from $W_{j''}$ (resp. $H_{j'}$ is disjoint from $H_{j''}$) for every $j' \neq j''$, and the total area of rectangles in $\bigcup W_i$ ($\bigcup H_i$ respectively) is at most $W \cdot OPT$. Thus, there exists a value $\overline{j}$ such that the total area of the elements in $W_{\overline{j}} \cup H_{\overline{j}}$ is at most $\dfrac{2 OPT\cdot W}{T} = \eps^k \cdot OPT\cdot W$. Choosing $\delta_h = y_{\overline{j}}$, $\mu_h = y_{\overline{j}+1}$, $\delta_w = x_{\overline{j}}$, $\mu_w = x_{\overline{j}+1}$ verifies all the conditions of the lemma.
\end{proof}

Function $f$ and constant $k$ will be chosen later. From now on, assume that $\delta_h, \delta_w, \mu_h$ and $\mu_w$ are chosen according to Lemma \ref{lem:mediumrectanglesarea}.  

\subsection{Next-Fit-Decreasing-Height (NFDH)}\label{sec:NFDH}

One of the most common algorithms to pack rectangles into a box of size $w\times h$ is Next-Fit-Decreasing-Height (NFDH). In this algorithm, the first step is to sort rectangles non-increasingly by height, say $h_1\geq h_2\geq \ldots \geq h_n$. Then, the first rectangle is packed in the bottom-left corner, and a shelf is defined of height $h_1$ and width $w$. The next rectangles are put in this shelf, next to each other and touching each other and the bottom of the shelf, until one does not fit, say the $i$-th one. At this point we define a new shelf above the first one, with height $h_i$. This process continues until all the rectangles are packed or the height of the next shelf does not fit inside the box.

This algorithm was studied by Coffman et al.~\cite{CGJT80} in the context of strip packing, in order to bound the obtained height when all the rectangles are packed into a strip. The result obtained can be summarized in the following lemma.

\begin{lemma}[Coffman et al.~\cite{CGJT80}]\label{lem:NFDH} Given a strip packing instance $(\R,W)$, algorithm NFDH gives a packing of height at most $h_{\max}(\R) + \frac{2a(\R)} {W}$. \end{lemma}

One important observation is that each horizontal shelf can be thought of as a vertical container. Another important property of the algorithm is that, if a given set of rectangles needs to be packed into a given bin, and all of them are relatively small compared to the dimensions of the bin, then NFDH is very efficient even in terms of area. This result is summarized in the following lemma.

\begin{lemma}[Coffman et al.~\cite{CGJT80}]\label{lem:NFDHarea} Given a set of rectangles with width at most $w$ and height at most $h$, if NFDH is used to pack these rectangles in a bin of width $a$ and height $b$, then the total used area in that bin is at least $(a-w)(b-h)$ (provided that there are enough rectangles so that NFDH never runs out of them). \end{lemma}

\subsection{Overview of the algorithm}

We next overview some of the basic results in \cite{NW16} that are required for our result. We define the constant $\gamma := \frac{\eps \delta_h}{2}$, and w.l.o.g. assume $\gamma\cdot OPT \in \mathbb{N}$.

Let us forget for a moment small rectangles $S$. We will pack all the remaining rectangles $L\cup H\cup T\cup V\cup M$ into a sufficiently small number of boxes embedded into the strip. By standard techniques, as in \cite{NW16}, it is then possible to pack $S$ (essentially using NFDH in a proper grid defined by the above boxes) while increasing the total height at most by $O(\eps)OPT$. See Section \ref{sec:smallpack} for more details on how to pack small rectangles.

The following lemma from \cite{NW16} allows one to round the heights and positions of rectangles of large enough height, without increasing much the height of the packing. 
\begin{lemma}\cite{NW16}\label{lem:verticalrounding}
There exists a feasible packing of height $OPT'\leq (1+\eps)OPT$ where: (1) the height of each rectangle in $L\cup T\cup V$ is rounded up to the closest integer multiple of $\gamma \cdot OPT$ and (2) their $x$-coordinates are as in the optimal solution and their $y$-coordinates are integer multiples of $\gamma \cdot OPT$.
\end{lemma}

 \arir{The proof  for $\gamma:=\eps \delta/2$ can be found in~\cite{NW16} and it is not affected by our slightly modified definition of $\gamma:=\eps \delta_h/2$. Should we explicitly cite \cite{NW16} for this?}

We next focus on rounded rectangle heights (i.e., implicitly replace $L\cup T\cup V$ by their rounded version) and on this slightly suboptimal solution of height $OPT'$.
 
The following lemma helps us to pack rectangles in $M$.
\begin{lemma}\label{lem:mediumrectanglesrepacking}
If $k$ in Lemma~\ref{lem:mediumrectanglesarea} is chosen sufficiently large, all the rectangles in $M$ can be packed in polynomial time into a box $B_{M,hor}$ of size $W \times O(\eps)OPT$ and a box $B_{M,ver}$ of size $(\frac{\gamma}{3} W)  \times (\alpha OPT)$. Furthermore, there is one such packing using $\frac{3\eps}{\mu_h}$ vertical containers in $B_{M,hor}$ and $\frac{\gamma}{3\mu_w}$ horizontal containers in $B_{M,ver}$.
\end{lemma}

\begin{proof} We first pack rectangles in $\mathcal{A} := \{R_i \in M : h_i \in (\mu_h OPT, \delta_h OPT )\}$ using NFDH into a strip of width $W$. From Lemma~\ref{lem:NFDH} we know that the height of the packing is at most $h_{\max}(\mathcal{A}) +  \frac{2\cdot a(\mathcal{A})}{W}$. Since $h_{max}(\mathcal{A}) \leq \delta_h OPT < \eps OPT$ and $a(\mathcal{A}) \le a(M) \le \eps^k \cdot OPT \cdot W \le \eps \cdot OPT \cdot W$, because of Lemma~\ref{lem:mediumrectanglesarea}, the resulting packing fits into a box $B_{M,hor}$ of size $W \times (3\eps \cdot OPT)$. As $h_i \ge \mu_h OPT$, the number of shelves used by NFDH is at most $\frac{3\eps}{\mu_h}$, and this also bounds the number of vertical containers needed. 
	
We next pack $\mathcal{A}' := M \setminus \mathcal{A}$ into a box $B_{M,ver}$ of size $(\frac{\gamma}{3} W)  \times (\alpha OPT)$. Recall that $\gamma := \frac{\eps \delta_h}{2}$. Note that, for each $R_i \in \mathcal{A}'$, we have $w_i \in (\mu_w W, \delta_w W)$ and $h_i \leq \alpha OPT$.  By ideally rotating the box and the rectangles by $90^\circ$, we can apply the NFDH algorithm. Lemma~\ref{lem:NFDH} implies that we can pack all the rectangles if the width of the box is at least $w_{\max}(\mathcal{A}') + \frac{2a(\mathcal{A}')}{\alpha OPT}$. Now observe that
	\[
	w_{\max}(\mathcal{A}') \leq \delta_w W = \frac{\eps \delta_h}{12} W = \frac{\gamma}{6}W
	\]
	and also, since $\alpha \geq 1/3$,
	\[
	\frac{2a(\mathcal{A}')}{\alpha OPT} \leq \frac{6a(\mathcal{A}')}{OPT} \leq 6 \eps^k W \le \frac{\gamma}{6}W,
	\]
	where the last inequality is true for any $k \geq \log_{1/\eps}{(36 / \gamma)}$. Similarly to the previous case, the number of shelves is at most $ \frac{\gamma}{3 \mu_w}$. Thus all the rectangles can be packed into at most $\frac{\gamma}{3 \mu_w}$  horizontal containers.
\end{proof}

We say that a rectangle $R_i$ is \emph{cut} by a box $B$ if both $R_i\setminus B$ and $B\setminus R_i$ are non-empty (considering both $R_i$ and $B$ as open regions with an implicit embedding on the plane). We say that a rectangle $R_i \in H$ (resp. $R_i \in T \cup V$) is \emph{nicely cut} by a box $B$ if $R_i$ is cut by $B$ and their intersection is a rectangular region of width $w_i$ (resp. height $h_i$). Intuitively, this means that an edge of $B$ cuts $R_i$ along its longest side (see Figure~\ref{fig_cut_rectangles}).

Now it remains to pack $L \cup H\cup T \cup V$: The following lemma, taken from \cite{NW16} modulo minor technical adaptations, describes an almost optimal packing of those rectangles. 

\begin{lemma}\label{lem:boxpartition}
There is an integer $K_B=(\frac{1}{\eps})(\frac{1}{\delta_w})^{O(1)}$ such that, assuming $\mu_h \le \frac{\eps \delta_w }{K_B}$, there is a partition of the region $B_{OPT'}:=[0,W]\times [0,OPT']$ into a set $\mathcal{B}$ of at most $K_B$ boxes and a packing of the rectangles in $L\cup T\cup V\cup H$ such that:
\begin{itemize}
\item each box has size equal to the size of some $R_i\in L$ (\emph{large box}), or has height at most $\delta_h OPT'$ (\emph{horizontal box}), or has width at most $\delta_w W$ (\emph{vertical box});
\item each $R_i \in L$ is contained into a large box of the same size;
\item each $R_i\in H$ is contained into a horizontal box or is cut by some box. Furthermore, the total area of horizontal cut rectangles is at most $W\cdot O(\eps)OPT'$;
\item each $R_i\in T\cup V$ is contained into a vertical box or is nicely cut by some vertical box. 
\end{itemize}
\end{lemma}

\begin{proof}
	We apply Lemma~3.2 in~\cite{NW16}, where we set the parameter $\delta$ to $\delta_w$. Recall that $\delta_w < \delta_h$; by requiring that $\mu_h < \delta_w$, and since rectangles with height in $[\delta_w, \delta_h)$ are in $M$, we have that $\{R_i \in \mathcal{R} \setminus M \,:\, w_i \geq \delta_h W \mbox{ and } h_i \geq \delta_h OPT \} = \{R_i \in \mathcal{R} \setminus M \,:\, w_i \geq \delta_w W \mbox{ and } h_i \geq \delta_h OPT \}$.
	
	Let $H_{cut}\subseteq H$ be the set of horizontal rectangles that are nicely cut by a box. Since rectangles in $H_{cut}$ satisfy $w_i \ge \delta_w W$, at most $\frac{2}{\delta_w}$ of them are nicely cut by a box, and there are at most $K_B$ boxes. Hence, their total area is at most $\frac{\mu_h OPT \cdot W \cdot 2K_B}{\delta_w}$, which is at most $2\eps \cdot OPT\cdot W$, provided that $\mu_h \le \eps\cdot \frac{\delta_w}{K_B}$. Since Lemma~3.2 in~\cite{NW16} implies that the area of the cut horizontal rectangles that are not nicely cut is at most $\eps OPT' \cdot W$, the total area of horizontal cut rectangles is at most $3 \eps OPT' \cdot W$.
\end{proof}

We denote the sets of vertical, horizontal, and large boxes by  $\mathcal{B}_V, \mathcal{B}_H$ and $\mathcal{B}_L$, respectively. Observe that $\mathcal{B}$ can be guessed in PPT. We next use $T_{cut}\subseteq T$ and $V_{cut}\subseteq V$ to denote tall and vertical cut rectangles in the above lemma, respectively. Let us also define $T_{box}=T\setminus T_{cut}$ and $V_{box}=V\setminus V_{cut}$.

Using standard techniques (see e.g. \cite{NW16}), we can pack all the rectangles excluding the ones contained in vertical boxes in a convenient manner. This is summarized in the following lemma.
\begin{lemma}\label{lem:structural_boxes}
Given $\mathcal{B}$ as in Lemma \ref{lem:boxpartition} and assuming $\mu_w \le \frac{\gamma \delta_h}{6K_B(1+\eps)}$, there exists a packing of $L\cup H\cup T \cup V$ such that: 
\begin{enumerate}
\item all the rectangles in $L$ are packed in $\mathcal{B}_L$;
\item all the rectangles in $H$ are packed in $\mathcal{B}_H$ plus an additional box $B_{H,cut}$ of size $W\times O(\eps)OPT$;
\item all the rectangles in $T_{cut}\cup T_{box}\cup V_{box}$ are packed as in Lemma \ref{lem:boxpartition};
\item all the rectangles in $V_{cut}$ are packed in an additional vertical box $B_{V,cut}$ of size $(\frac{\gamma}{3} W)  \times (\alpha OPT)$.
\end{enumerate}
\end{lemma}

\begin{proof}
	Note that there are at most $1/(\delta_w \delta_h)$ rectangles in $L$ and at most $4K_B$ rectangles in $T_{cut}$, since at most $2$ tall rectangles can be nicely cut by the left (resp. right) side of each box; this is enough to prove points~(1) and~(3).
	
	Thanks to Lemma \ref{lem:boxpartition}, the total area of horizontal cut rectangles is at most $O(\eps OPT' \cdot W)$. By Lemma~\ref{lem:NFDH}, we can remove them from the packing and pack them in the additional box $B_{H, cut}$ using NFDH algorithm, proving point~(2).
	
	At most $\frac{2(1+\eps)}{\delta_h}$ rectangles in $V$ can be nicely cut by a box; thus, in total there are at most $\frac{2K_B(1+\eps)}{\delta_h}$ nicely cut vertical rectangles. Since the width of each vertical rectangle is at most $\mu_w W$, they can be removed from the packing and placed in $B_{V, cut}$, piled side by side, as long as $\frac{2K_B(1+\eps)}{\delta_h} \cdot \mu_w W\leq \frac{\gamma}{3} W$, which is equivalent to $\mu_w \le \frac{\gamma \delta_h}{6K_B(1+\eps)}$. This proves point~(4).
\end{proof}

We will pack all the rectangles (essentially) as in \cite{NW16}, with the exception of $T_{box}\cup V_{box}$ where we exploit a refined approach. This is the technical heart of this paper, and it is discussed in the next section. 

\section{A repacking lemma}
\label{sec:repack}
We next describe how to pack rectangles in $T_{box}\cup V_{box}$. In order to highlight our contribution, we first describe how the approach by Nadiradze and Wiese \cite{NW16} works.

It is convenient to assume that all the rectangles in $V_{box}$ are sliced vertically into sub-rectangles of width $1$ each\footnote{For technical reasons, slices have width $1/2$ in~\cite{NW16}. For our algorithm, slices of width $1$ suffice.}. Let $V_{sliced}$ be such \emph{sliced} rectangles. We will show how to pack all the rectangles in $T_{box}\cup V_{sliced}$ into a constant number of sub-boxes. Using standard techniques it is then possible to pack $V_{box}$ into the space occupied by $V_{sliced}$ plus an additional box $B_{V,round}$ of size $(\frac{\gamma}{3}W)\times \alpha OPT$.
See Lemma~\ref{lem:fractional_to_integral} for more details.

We next focus on a specific vertical box $\overline{B}$, say of size $\overline{w}\times \overline{h}$ (see Figure \ref{fig_pseudo-rectangles1}). Let $\overline{T}_{cut}$ be the tall rectangles cut by $\overline{B}$. Observe that there are at most $4$ such rectangles ($2$ on the left/right side of $\overline{B}$).
The rectangles in $\overline{T}_{cut}$ are packed as in Lemma \ref{lem:structural_boxes}.
Let also $\overline{T}$ and $\overline{V}$ be the tall rectangles and sliced vertical rectangles, respectively, originally packed completely inside  $\overline{B}$.

They show that it is possible to pack $\overline{T}\cup\overline{V}$ into a constant size set $\overline{\mathcal{S}}$ of sub-boxes contained inside $\overline{B}-\overline{T}_{cut}$, plus an additional box $\overline{D}$ of size $\overline{w}\times (1+\eps-2\alpha)\overline{h}$. Here $\overline{B}-\overline{T}_{cut}$ denotes the region inside $\overline{B}$ not contained in $\overline{T}_{cut}$.
In more detail, they start by considering each rectangle $R_i\in \overline{T}$. Since $\alpha\geq\frac{1}{3}$ by assumption, one of the regions above or below $R_i$ cannot contain another tall rectangle in $\overline{T}$, say the first case applies (the other one being symmetric). Then $R_i$ is moved up so that its top side overlaps with the top boundary of $\overline{B}$. The sliced rectangles in $\overline{V}$ that are covered this way are shifted right below $R$ (note that there is enough free space by construction). At the end of the process all the rectangles in $\overline{T}$ touch at least one of the top and bottom side of $\overline{B}$ (see Figure~\ref{fig_pseudo-rectangles2}). Note that no rectangle is discarded up to this point.

Next, we partition the space inside $\overline{B}-(\overline{T}\cup \overline{T}_{cut})$ into maximal height unit-width vertical stripes. We call each such stripe a \emph{free rectangle} if both its top and bottom side overlap with the top or bottom side of some rectangle in $\overline{T}\cup \overline{T}_{cut}$, and otherwise a \emph{pseudo rectangle} (see Figure \ref{fig_pseudo-rectangles3}). We define the $i$-th free rectangle to be the free rectangle contained in stripe $[i-1,i]\times [0,\overline{h}]$.

Note that all the free rectangles are contained in a rectangular region of width $\overline{w}$ and height at most $\overline{h}-2\alpha OPT \le \overline{h}-2\alpha \frac{OPT'}{1+\eps}
\le \overline{h}(1-\frac{2\alpha}{1+\eps}) \le \overline{h}(1+\eps-2\alpha)$ contained in the central part of $\overline{B}$.
Let $\overline{V}_{disc}$ be the set of  (sliced vertical) rectangles contained in the free rectangles.
Rectangles in $\overline{V}_{disc}$ can be obviously packed inside $\overline{D}$. For each corner $Q$ of the box $\overline{B}$, we consider the maximal rectangular region that has $Q$ as a corner and only contains pseudo rectangles whose top/bottom side overlaps with the bottom/top side of a rectangle in $\overline{T}_{cut}$; there are at most $4$ such non-empty regions, and for each of them we define a \emph{corner sub-box}, and we call the set of such sub-boxes $\overline{B}_{corn}$ (see Figure \ref{fig_pseudo-rectangles3}). The final step of the algorithm is to rearrange horizontally the pseudo/tall rectangles so that pseudo/tall rectangles of the same height are grouped together \emph{as much as possible} (modulo some technical details). The rectangles in $\overline{B}_{corn}$ are not moved.
The \emph{sub-boxes} are induced by maximal consecutive subsets of pseudo/tall rectangles of the same height touching the top (resp., bottom) side of $\overline{B}$
(see Figure \ref{fig_pseudo-rectangles4}).
We crucially remark that, by construction, the height of each sub-box (and of $\overline{B}$) is a multiple of $\gamma OPT$.

By splitting each discarded box $\overline{D}$ into two halves $\overline{B}_{disc,top}$ and $\overline{B}_{disc,bot}$, and replicating the packing of boxes inside $B_{OPT'}$, it is possible to pack all the discarded boxes into two boxes $B_{disc,top}$ and $B_{disc,bot}$, both of size $\frac{W}{2}\times (1+\eps-2\alpha)OPT'$. 

A feasible packing of boxes (and hence of the associated rectangles) of height $(1+\max\{\alpha,2(1-2\alpha)\}+O(\eps))OPT$ is then obtained as follows. We first pack $B_{OPT'}$ at the base of the strip, and then on top of it we pack $B_{M,hor}$, two additional boxes $B_{H, round}$ and $B_{H, cut}$ (which will be used to repack the horizontal items; see Section~\ref{lem:structural_boxes} for details), and a box $B_S$ (which will be used to pack some of the small items). The latter $4$ boxes all have width $W$ and height $O(\eps OPT')$. On the top right of this packing we place $B_{disc,top}$ and $B_{disc,bot}$, one on top of the other. Finally, we pack $B_{M,ver}$, $B_{V,cut}$ and $B_{V,round}$ on the top left, one next to the other. See Figure~\ref{fig_packing_NW} for an illustration.
The height is minimized for $\alpha=\frac{2}{5}$, leading to a $7/5+O(\eps)$ approximation.

The main technical contribution of this paper is to show how it is possible to repack a subset of $\overline{V}_{disc}$ into the \emph{free} space inside $\overline{B}_{cut}:=\overline{B}-\overline{T}_{cut}$ not occupied by sub-boxes, so that the residual sliced rectangles can be packed into a single discarded box $\overline{B}_{disc}$ of size $(1-\gamma)\overline{w}\times (1+\eps-2\alpha)\overline{h}$ (\emph{repacking lemma}). See Figure 
\ref{fig_pseudo-rectangles5}. This apparently minor saving is indeed crucial: with the same approach as above all the discarded sub-boxes $\overline{B}_{disc}$ can be packed into a single \emph{discarded box} $B_{disc}$ of size $(1-\gamma)W\times (1+\eps-2\alpha)OPT'$. Therefore, we can pack all the previous boxes as before, and $B_{disc}$ on the top right. Indeed, the total width of $B_{M,ver}$, $B_{V,cut}$ and $B_{V,round}$ is at most $\gamma W$ for a proper choice of the parameters. See Figure \ref{fig_our_packing} for an illustration. Altogether the resulting packing has height $(1+\max\{\alpha,1-2\alpha\}+O(\eps))OPT$. This is minimized for $\alpha=\frac{1}{3}$, leading to the claimed $4/3+O(\eps)$ approximation. 

\begin{figure}
	\captionsetup[subfigure]{justification=centering}
	\hspace{-10pt}
	\begin{subfigure}[b]{.32\textwidth}
		\centering
\resizebox{!}{3.5cm}{
\begin{tikzpicture}

\draw[fill=darkgray] (-1,0.5) rectangle (1.5,4);
\draw[fill=darkgray] (6,1) rectangle (7.5,4);
\draw[fill=darkgray] (6,4.5) rectangle (8,7.5);


\draw[fill=gray] (2,1.5) rectangle (3,5.5);
\draw[fill=gray] (3,0.5) rectangle (4.5,3.5);
\draw[fill=gray] (4.5,1) rectangle (6,4.5);
\draw[fill=gray] (0,4) rectangle (1,7);
\draw[fill=gray] (3,4) rectangle (3.5,7);
\draw[fill=gray] (3.5,3.5) rectangle (4.5,7.5);
\draw[fill=gray] (4.5,5) rectangle (6,8);


\draw[fill=lightgray] (0.5,0) rectangle (1,0.5);
\draw[fill=lightgray] (1.5,0.5) rectangle (2,1.5);
\draw[fill=lightgray] (2,0.5) rectangle (2.5,1.5);
\draw[fill=lightgray] (2.5,0) rectangle (3,1);
\draw[fill=lightgray] (3.5,0) rectangle (4,0.5);
\draw[fill=lightgray] (4.5,0) rectangle (5,1);
\draw[fill=lightgray] (5,0) rectangle (5.5,1);
\draw[fill=lightgray] (6,0.5) rectangle (6.5,1);
\draw[fill=lightgray] (4,0) rectangle (4.5,0.5);
\draw[fill=lightgray] (4,7.5) rectangle (4.5,8);
\draw[fill=lightgray] (6.5,4) rectangle (7,4.5);


\draw[fill=lightgray] (0,7) rectangle (0.5,8);
\draw[fill=lightgray] (1,6.5) rectangle (1.5,8);
\draw[fill=lightgray] (1,4.5) rectangle (1.5,5.5);
\draw[fill=lightgray] (1.5,6.5) rectangle (2,8);
\draw[fill=lightgray] (1.5,3) rectangle (2,5);
\draw[fill=lightgray] (2,6) rectangle (2.5,7);
\draw[fill=lightgray] (2.5,6.5) rectangle (3,7.5);
\draw[fill=lightgray] (3.5,7.5) rectangle (4,8);


\draw[thick] (0,0) rectangle (7,8);
\draw[color=white] (0,-0.25) rectangle (0,-0.1);
\draw[color=white] (0,8.3) rectangle (0,8.1);

\end{tikzpicture}}
		\caption{Original packing in a vertical box $\overline{B}$ after removing $V_{cut}$. Gray rectangles correspond to $\overline{T}$, dark gray ones to $\overline{T}_{cut}$ and light gray ones to $\overline{V}$.}
		\label{fig_pseudo-rectangles1}
	\end{subfigure}%
	\begin{subfigure}[b]{.32\textwidth}
		\centering
\resizebox{!}{3.5cm}{
\begin{tikzpicture}

\draw[fill=darkgray] (-1,0.5) rectangle (1.5,4);
\draw[fill=darkgray] (6,1) rectangle (7.5,4);
\draw[fill=darkgray] (6,4.5) rectangle (8,7.5);


\draw[fill=gray] (2,0) rectangle (3,4);
\draw[fill=gray] (3,0) rectangle (4.5,3);
\draw[fill=gray] (4.5,0) rectangle (6,3.5);
\draw[fill=gray] (0,5) rectangle (1,8);
\draw[fill=gray] (3,5) rectangle (3.5,8);
\draw[fill=gray] (3.5,4) rectangle (4.5,8);
\draw[fill=gray] (4.5,5) rectangle (6,8);


\draw[fill=lightgray] (0.5,0) rectangle (1,0.5);
\draw[fill=lightgray] (1.5,0.5) rectangle (2,1.5);
\draw[fill=lightgray] (2,4.5) rectangle (2.5,5.5);
\draw[fill=lightgray] (2.5,4) rectangle (3,5);
\draw[fill=lightgray] (3.5,3) rectangle (4,3.5);
\draw[fill=lightgray] (4.5,3.5) rectangle (5,4.5);
\draw[fill=lightgray] (5,3.5) rectangle (5.5,4.5);
\draw[fill=lightgray] (4,3) rectangle (4.5,3.5);
\draw[fill=lightgray] (4,3.5) rectangle (4.5,4);
\draw[fill=lightgray] (6.5,4) rectangle (7,4.5);
\draw[fill=lightgray] (6,0.5) rectangle (6.5,1);


\draw[fill=lightgray] (0,4) rectangle (0.5,5);
\draw[fill=lightgray] (1,6.5) rectangle (1.5,8);
\draw[fill=lightgray] (1,4.5) rectangle (1.5,5.5);
\draw[fill=lightgray] (1.5,6.5) rectangle (2,8);
\draw[fill=lightgray] (1.5,3) rectangle (2,5);
\draw[fill=lightgray] (2,6) rectangle (2.5,7);
\draw[fill=lightgray] (2.5,6.5) rectangle (3,7.5);
\draw[fill=lightgray] (3.5,3.5) rectangle (4,4);


\draw[thick] (0,0) rectangle (7,8);
\draw[color=white] (0,-0.25) rectangle (0,-0.1);
\draw[color=white] (0,8.3) rectangle (0,8.1);

\end{tikzpicture}}
		\caption{Rectangles in $\overline{T}$ are shifted vertically so that \\they touch either the top \\or the bottom of box $\overline{B}$, shifting also slices in $\overline{V}$ accordingly.}
		\label{fig_pseudo-rectangles2}
	\end{subfigure}
	\begin{subfigure}[b]{.36\textwidth}
		\centering
\resizebox{!}{3.5cm}{
\begin{tikzpicture}


\draw[thick, dashed] (0.5,4) -- (0.5,5);
\draw[thick, dashed] (1.5,4) -- (1.5,8);
\draw[thick, dashed] (2.5,4) -- (2.5,8);
\draw[thick, dashed] (3.5,3) -- (3.5,4);
\draw[thick, dashed] (4.5,3.5) -- (4.5,4);
\draw[thick, dashed] (5.5,3.5) -- (5.5,5);
\draw[thick, dashed] (6.5,4) -- (6.5,4.5);
\draw[thick, dashed] (1,4) -- (1,5);
\draw[thick, dashed] (2,4) -- (2,8);
\draw[thick, dashed] (3,4) -- (3,5);
\draw[thick, dashed] (4,3) -- (4,4);
\draw[thick, dashed] (5,3.5) -- (5,5);
\draw[thick, dashed] (6,4) -- (6,4.5);


\fill [pattern color=lightgray, pattern = crosshatch] (1.5,0) rectangle (2,8);
\fill [pattern color=lightgray, pattern = crosshatch] (2,4) rectangle (2.5,8);
\fill [pattern color=lightgray, pattern = crosshatch] (2.5,4) rectangle (3,8);

\fill [pattern color=lightgray, pattern = crosshatch] (1,4) rectangle (1.5,8);


\draw (0,0) rectangle (1.5,0.5);
\fill[thick, pattern = north east lines] (0,0) rectangle (1.5,0.5);
\draw (6,0) rectangle (7,1);
\fill[thick, pattern = north east lines] (6,0) rectangle (7,1);
\draw (6,7.5) rectangle (7,8);
\fill[thick, pattern = north east lines] (6,7.5) rectangle (7,8);


\draw[thick] (1.5,0) -- (1.5,4) -- (0,4) -- (0,5) -- (1,5) -- (1,8) -- (3,8) -- (3,5) -- (3.5,5) -- (3.5,4) -- (4.5,4) -- (4.5,5) -- (6,5) -- (6,4.5) -- (7,4.5) -- (7,4) -- (6,4) -- (6,3.5) -- (4.5,3.5) -- (4.5,3) -- (3,3) -- (3,4) -- (2,4) -- (2,0) -- (1.5,0);


\draw[dashed] (0,0) --(7.5,0);
\draw[dashed] (0,3) --(7.5,3);
\draw[dashed] (0,5) --(7.5,5);
\draw[dashed] (0,8) --(7.5,8);

\draw (7.5,0) node[anchor=west] {\Large $0$};
\draw (7.5,3) node[anchor=west] {\Large $\frac{\alpha}{1+\eps} \overline{h}$};
\draw (7.5,5) node[anchor=west] {\Large $(1-\frac{\alpha}{1+\epsilon}) \overline{h}$};
\draw (7.5,8) node[anchor=west] {\Large $\overline{h}$};

\end{tikzpicture}}
		\caption{Classification in $\overline{B} - (\overline{T} \cup \overline{T}_{cut})$. Crosshatched stripes correspond to pseudo rectangles, empty stripes to free rectangles, and dashed regions to corner sub-boxes.}
		\label{fig_pseudo-rectangles3}
	\end{subfigure}
	
	\begin{subfigure}[b]{.5\textwidth}
		\centering
\resizebox{!}{4.5cm}{
\begin{tikzpicture}

\draw[fill=darkgray] (-1,0.5) rectangle (1.5,4);
\draw[fill=darkgray] (6,1) rectangle (7.5,4);
\draw[fill=darkgray] (6,4.5) rectangle (8,7.5);


\draw[fill=gray] (4,0) rectangle (5.5,4);
\draw[fill=gray] (1.5,0) rectangle (2.5,3);
\draw[fill=gray] (2.5,0) rectangle (4,3.5);
\draw[fill=gray] (5,5) rectangle (5.5,8);
\draw[fill=gray] (4,5) rectangle (5,8);
\draw[fill=gray] (1.5,4) rectangle (2.5,8);
\draw[fill=gray] (2.5,5) rectangle (4,8);


\fill [pattern color=lightgray, pattern = crosshatch] (5.5,0) rectangle (6,8);
\fill [pattern color=lightgray, pattern = crosshatch] (1,4) rectangle (1.5,8);
\fill [pattern color=lightgray, pattern = crosshatch] (0.5,4) rectangle (1,8);
\fill [pattern color=lightgray, pattern = crosshatch] (0,4) rectangle (0.5,8);


\fill[pattern = north east lines] (0,0) rectangle (1.5,0.5);
\fill[pattern = north east lines] (6,0) rectangle (7,1);
\fill[pattern = north east lines] (6,7.5) rectangle (7,8);


\draw[thick, dashed] (0.5,4) -- (0.5,8);
\draw[thick, dashed] (1,4) -- (1,8);
\draw[thick, dashed] (2,3) -- (2,4);
\draw[thick, dashed] (2.5,3.5) -- (2.5,4);
\draw[thick, dashed] (3,3.5) -- (3,5);
\draw[thick, dashed] (3.5,3.5) -- (3.5,5);
\draw[thick, dashed] (4,4) -- (4,5);
\draw[thick, dashed] (4.5,4) -- (4.5,5);
\draw[thick, dashed] (5,4) -- (5,5);
\draw[thick, dashed] (5.5,4) -- (5.5,5);
\draw[thick, dashed] (6,4) -- (6,4.5);
\draw[thick, dashed] (6.5,4) -- (6.5,4.5);


\draw[thick] (0,0) rectangle (7,8);


\draw[ultra thick] (0,0) rectangle (1.5,0.5);
\draw[ultra thick] (6,0) rectangle (7,1);
\draw[ultra thick] (7,7.5) rectangle (7,8);

\draw[ultra thick] (-1,0.5) rectangle (1.5,4);
\draw[ultra thick] (6,1) rectangle (7.5,4);
\draw[ultra thick] (6,4.5) rectangle (8,7.5);

\draw[ultra thick] (0,4) rectangle (1.5,8);
\draw[ultra thick] (5.5,0) rectangle (6,8);

\draw[ultra thick] (1.5,0) rectangle (2.5,3);
\draw[ultra thick] (2.5,0) rectangle (4,3.5);
\draw[ultra thick] (4,0) rectangle (5.5,4);
\draw[ultra thick] (1.5,4) rectangle (2.5,8);
\draw[ultra thick] (2.5,5) rectangle (5.5,8);


\draw[fill=lightgray] (8.5,2) rectangle (9,3);


\draw[fill=lightgray] (8.5,4.5) rectangle (9,5);
\draw[fill=lightgray] (8.5,5) rectangle (9,5.5);
\draw[fill=lightgray] (9,4.5) rectangle (9.5,5);
\draw[fill=lightgray] (9,5) rectangle (9.5,5.5);
\draw[fill=lightgray] (9.5,5) rectangle (10,6);
\draw[fill=lightgray] (10,5) rectangle (10.5,6);
\draw[fill=lightgray] (11.5,5.5) rectangle (12,6);


\draw[dashed] (9,1) -- (9,3);
\draw[dashed] (9.5,1) -- (9.5,3);
\draw[dashed] (10,1) -- (10,3);
\draw[dashed] (10.5,1) -- (10.5,3);
\draw[dashed] (11,1) -- (11,3);
\draw[dashed] (11.5,1) -- (11.5,3);

\draw[dashed] (9,4.5) -- (9,6.5);
\draw[dashed] (9.5,4.5) -- (9.5,6.5);
\draw[dashed] (10,4.5) -- (10,6.5);
\draw[dashed] (10.5,4.5) -- (10.5,6.5);
\draw[dashed] (11,4.5) -- (11,6.5);
\draw[dashed] (11.5,4.5) -- (11.5,6.5);


\draw (8.5,1) rectangle (12,3);
\draw (8.5,4.5) rectangle (12,6.5);


\draw [thick,decorate,decoration={mirror, brace,amplitude=8pt}] 
(8.5,1) -- (12,1); 
\draw (10.25,0.5) node [anchor = north] {\Large $\frac{1}{2}\overline{w}$};

\draw [thick,decorate,decoration={mirror, brace,amplitude=8pt}] 
(12,1) -- (12,3); 
\draw (12.25,2) node [rotate=90, anchor = north] {\Large $(1+\eps-2\alpha)\overline{h}$};

\draw [thick,decorate,decoration={mirror, brace,amplitude=8pt}] 
(12,4.5) -- (12,6.5); 
\draw (12.25,5.5) node [rotate=90, anchor = north] {\Large $(1+\eps-2\alpha)\overline{h}$};

\draw (10.25,3) node [anchor=south] {\Large $\overline{B}_{disc, bot}$};

\draw (10.25,6.5) node [anchor=south] {\Large $\overline{B}_{disc, top}$};

\draw[color=white] (0,-2.05) rectangle (1,-1);

\end{tikzpicture}}
		\caption{Rearrangement of pseudo and tall rectangles to get $O_\eps(1)$ sub-boxes, and additional packing of $\overline{V}_{disc}$ as in \cite{NW16}.}
		\label{fig_pseudo-rectangles4}
	\end{subfigure}%
	\hspace{-2pt}
	\begin{subfigure}[b]{.5\textwidth}
		\centering
\resizebox{!}{4.5cm}{
\begin{tikzpicture}

\draw[fill=darkgray] (-1,0.5) rectangle (1.5,4);
\draw[fill=darkgray] (6,1) rectangle (7.5,4);
\draw[fill=darkgray] (6,4.5) rectangle (8,7.5);


\draw[fill=gray] (4,0) rectangle (5.5,4);
\draw[fill=gray] (1.5,0) rectangle (2.5,3);
\draw[fill=gray] (2.5,0) rectangle (4,3.5);
\draw[fill=gray] (5,5) rectangle (5.5,8);
\draw[fill=gray] (4,5) rectangle (5,8);
\draw[fill=gray] (1.5,4) rectangle (2.5,8);
\draw[fill=gray] (2.5,5) rectangle (4,8);


\fill [pattern color=lightgray, pattern = crosshatch] (5.5,0) rectangle (6,8);
\fill [pattern color=lightgray, pattern = crosshatch] (1,4) rectangle (1.5,8);
\fill [pattern color=lightgray, pattern = crosshatch] (0.5,4) rectangle (1,8);
\fill [pattern color=lightgray, pattern = crosshatch] (0,4) rectangle (0.5,8);


\draw[fill=lightgray] (3.5,3.5) rectangle (4,4);
\draw[fill=lightgray] (3.5,4) rectangle (4,4.5);
\draw[fill=lightgray] (4,4) rectangle (4.5,4.5);
\draw[fill=lightgray] (4,4.5) rectangle (4.5,5);
\draw[fill=lightgray] (6.5,4) rectangle (7,4.5);


\fill[pattern = north east lines] (0,0) rectangle (1.5,0.5);
\fill[pattern = north east lines] (6,0) rectangle (7,1);
\fill[pattern = north east lines] (6,7.5) rectangle (7,8);


\draw[thick, dashed] (0.5,4) -- (0.5,8);
\draw[thick, dashed] (1,4) -- (1,8);
\draw[thick, dashed] (2,3) -- (2,4);
\draw[thick, dashed] (2.5,3.5) -- (2.5,4);
\draw[thick, dashed] (3,3.5) -- (3,5);
\draw[thick, dashed] (3.5,3.5) -- (3.5,5);
\draw[thick, dashed] (4,4) -- (4,5);
\draw[thick, dashed] (4.5,4) -- (4.5,5);
\draw[thick, dashed] (5,4) -- (5,5);
\draw[thick, dashed] (5.5,4) -- (5.5,5);
\draw[thick, dashed] (6,4) -- (6,4.5);
\draw[thick, dashed] (6.5,4) -- (6.5,4.5);


\draw[thick] (0,0) rectangle (7,8);


\draw[ultra thick] (0,0) rectangle (1.5,0.5);
\draw[ultra thick] (6,0) rectangle (7,1);
\draw[ultra thick] (7,7.5) rectangle (7,8);

\draw[ultra thick] (-1,0.5) rectangle (1.5,4);
\draw[ultra thick] (6,1) rectangle (7.5,4);
\draw[ultra thick] (6,4.5) rectangle (8,7.5);

\draw[ultra thick] (0,4) rectangle (1.5,8);
\draw[ultra thick] (5.5,0) rectangle (6,8);

\draw[ultra thick] (1.5,0) rectangle (2.5,3);
\draw[ultra thick] (2.5,0) rectangle (4,3.5);
\draw[ultra thick] (4,0) rectangle (5.5,4);
\draw[ultra thick] (1.5,4) rectangle (2.5,8);
\draw[ultra thick] (2.5,5) rectangle (5.5,8);


\draw[dashed] (9,3) -- (9,5);
\draw[dashed] (9.5,3) -- (9.5,5);
\draw[dashed] (10,3) -- (10,5);
\draw[dashed] (10.5,3) -- (10.5,5);
\draw[dashed] (11,3) -- (11,5);
\draw[dashed] (11.5,3) -- (11.5,5);


\draw[fill=lightgray] (8.5,4) rectangle (9,5);
\draw[fill=lightgray] (10,3.5) rectangle (10.5,4.5);
\draw[fill=lightgray] (10.5,3.5) rectangle (11,4.5);


\draw (8.5,3) rectangle (11.5,5);


\draw[ultra thick, ->] (1.25,-0.9) -- (1.25,-0.1);
\draw[ultra thick, ->] (1.75,-0.9) -- (1.75,-0.1);
\draw[ultra thick, ->] (2.25,-0.9) -- (2.25,-0.1);
\draw[ultra thick, ->] (2.75,-0.9) -- (2.75,-0.1);
\draw[ultra thick, ->] (3.75,-0.9) -- (3.75,-0.1);
\draw[ultra thick, ->] (4.25,-0.9) -- (4.25,-0.1);
\draw[ultra thick, ->] (6.25,-0.9) -- (6.25,-0.1);
\draw[ultra thick, ->] (6.75,-0.9) -- (6.75,-0.1);

\draw (3.5,-1) node[anchor=north] {\Large $\ge \gamma \overline{w}$};
\draw (3.5,-1.45) node[anchor=north] {\Large good indexes};

\draw [thick,decorate,decoration={mirror, brace,amplitude=8pt}] 
(8.5,3) -- (11.5,3); 
\draw (10,2.5) node [anchor = north] {\Large $\le (1-\gamma)\overline{w}$};

\draw [thick,decorate,decoration={mirror, brace,amplitude=8pt}] 
(11.5,3) -- (11.5,5); 
\draw (11.75,4) node [rotate=90, anchor = north] {\Large $\left(1+\eps-2\alpha\right)\overline{h}$};

\draw (10,5) node[anchor = south] {\Large $\overline{B}_{disc}$};

\end{tikzpicture}}
		\caption{Our refined repacking of $\overline{V}_{disc}$ according to Lemma \ref{lem:repack}: some vertical slices are repacked in the free space.}
		\label{fig_pseudo-rectangles5}
	\end{subfigure}	
	\caption{Creation of pseudo rectangles, how to get constant number of sub-boxes and repacking of vertical slices in a vertical box $\overline{B}$.}
	\label{fig_pseudo-rectangles}
\end{figure}
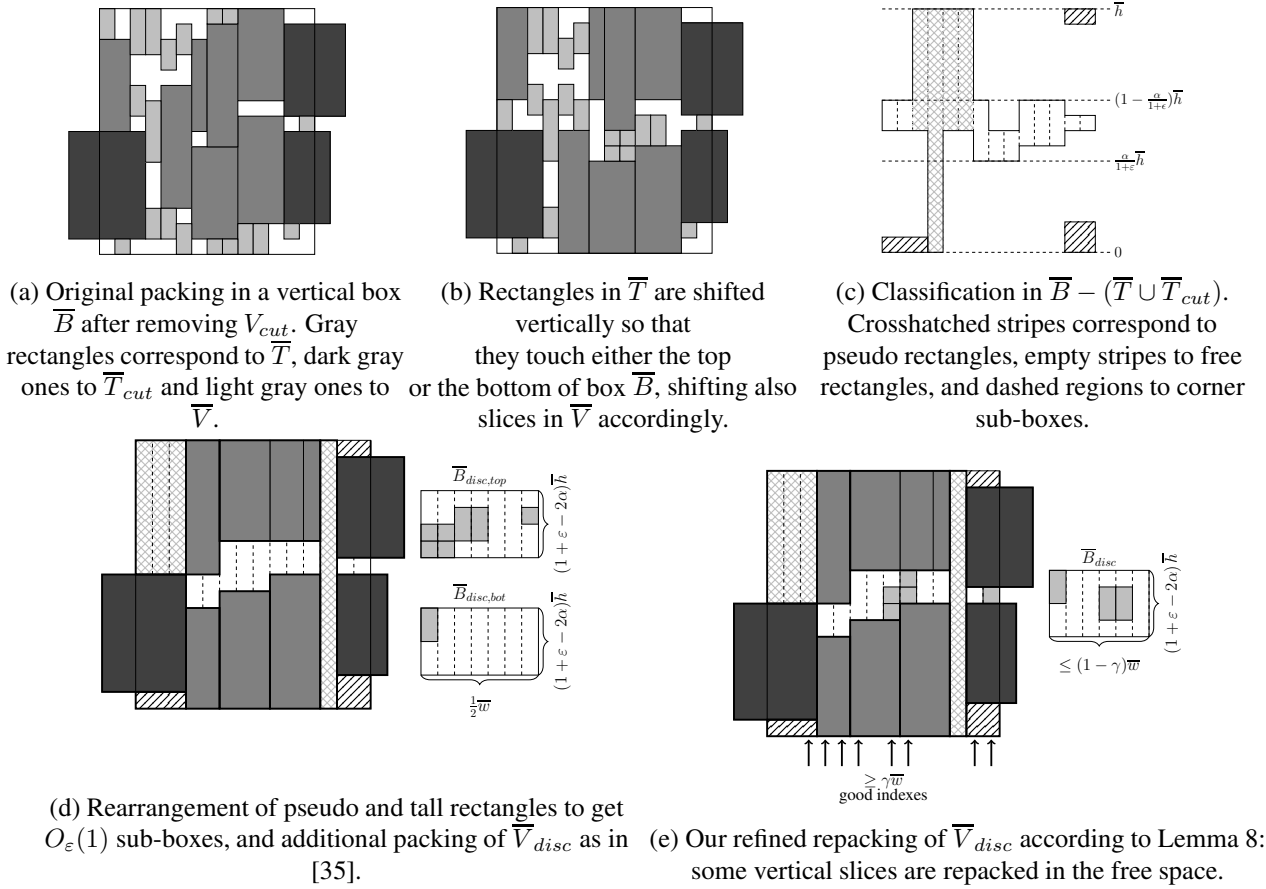

It remains to prove our repacking lemma.
\begin{lemma}[Repacking Lemma]\label{lem:repack}
Consider a partition of $\overline{D}$ into $\overline{w}$ unit-width vertical stripes. There is a subset of at least $\gamma \overline{w}$ such stripes so that the corresponding sliced vertical rectangles $\overline{V}_{repack}$ can be repacked inside $\overline{B}_{cut}=\overline{B}-\overline{T}_{cut}$ in the space not occupied by sub-boxes.
\end{lemma}
\begin{proof}
Let $f(i)$ denote the height of the $i$-th free rectangle, where for notational convenience we introduce a degenerate free rectangle of height $f(i)=0$ whenever the stripe $[i-1,i]\times [0,\overline{h}]$ inside $\overline{B}$ does not contain any free  rectangle. This way we have precisely $\overline{w}$ free rectangles. We remark that free rectangles are defined before the horizontal rearrangement of tall/pseudo rectangles, and the consequent definition of sub-boxes. 

Recall that sub-boxes contain tall and pseudo rectangles.
Now consider the area in $\overline{B}_{cut}$ not occupied by sub-boxes. Note that this area is  contained in the central region of height $\overline{h}(1 - \frac{2\alpha}{1 + \eps})$. Partition this area into maximal-height unit-width vertical stripes as before (\emph{newly free rectangles}). Let $g(i)$ be the height of the $i$-th newly free rectangle, where again we let $g(i)=0$ if the stripe $[i-1,i]\times [0,\overline{h}]$ does not contain any (positive area) free region. Note that, since  tall and pseudo rectangles are only shifted horizontally in the rearrangement, it must be the case that:
\[
\sum_{i=1}^{\overline{w}}f(i)=\sum_{i=1}^{\overline{w}}g(i).
\]
Let $G$ be the (\emph{good}) indexes where $g(i)\geq f(i)$, and $\overline{G}=\{1,\ldots,\overline{w}\}\setminus G$ 
be the \emph{bad} indexes with $g(i)<f(i)$. Observe that for each $i\in G$, it is possible to pack the $i$-th free rectangle inside the $i$-th newly free rectangle, therefore freeing a unit-width vertical strip inside $\overline{D}$. Thus it is sufficient to show that $|G|\geq \gamma \overline{w}$.

Observe that, for $i\in \overline{G}$, $f(i)-g(i)\geq \gamma OPT\geq \gamma \frac{\overline{h}}{1 + \eps}$: indeed, both $f(i)$ and $g(i)$ must be multiples of $\gamma OPT$ since they correspond to the height of $\overline{B}$ minus the height of one or two tall/pseudo rectangles. On the other hand, for any index $i$, $g(i)-f(i)\leq g(i)\leq (1-\frac{2\alpha}{1+\eps})\overline{h}$, by the definition of $g$.
Altogether
\[
(1-\frac{2\alpha}{1 + \eps})\overline{h} \cdot |G| \geq \sum_{i\in G}(g(i)-f(i)) = \sum_{i\in \overline{G}}(f(i)-g(i))\geq  \frac{\gamma \overline{h}}{1 + \eps} \cdot |\overline{G}| = \frac{\gamma \overline{h}}{1 + \eps} \cdot (\overline{w}-|G|) 
\]
We conclude that $|G|\geq \frac{\gamma}{1+\eps-2\alpha+\gamma}\overline{w}$. The claim follows since by assumption $\alpha>\eps\geq \gamma$. 
\end{proof}

\section{A refined structural lemma}\label{sec:structural_lemma}

The original algorithm in \cite{NW16} uses standard LP-based techniques, as in \cite{KenyonR00}, to pack the horizontal rectangles. We can avoid that via a refined structural lemma: here boxes and sub-boxes are further partitioned into vertical (resp., horizontal) containers. Rectangles are then packed into such containers as mentioned earlier: one next to the other from left to right (resp., bottom to top). Containers define a multiple knapsack instance, that can be solved optimally in PPT via dynamic programming. This approach has two main advantages:
\begin{itemize}\itemsep0pt
\item It leads to a simpler algorithm.
\item It can be easily adapted to the case with rotations, as discussed in Section~\ref{sec:rot}.
\end{itemize}

The goal of this section is to prove the following lemma that summarizes the aforementioned properties.

\begin{lemma}\label{lem:structural_containers}
By choosing $\alpha = 1/3$, there is an integer $K_F \leq \left(\frac{1}{\eps \delta_w}\right)^{O(1/(\delta_w \eps))}$ such that, assuming $\mu_h \leq \frac{\eps}{K_F}$ and $\mu_w \leq \frac{\gamma}{3K_F}$, there is a packing of $\mathcal{R} \setminus S$ in the region $[0, W] \times [0, (4/3 + O(\eps))OPT']$ with the following properties:
\begin{itemize}
\item All the rectangles in $\mathcal{R} \setminus S$ are contained in $K_{TOTAL} = O_\eps(1)$ horizontal or vertical containers, such that each of these containers is either contained in or disjoint from $\mathcal{B}_{OPT'}$;
\item At most $K_F$ containers are contained in $\mathcal{B}_{OPT'}$, and their total area is at most $a(\mathcal{R} \setminus S)$.
\end{itemize}
\end{lemma}

Given a set $\mathcal{M} \subseteq \mathcal{R}$ of rectangles, we define $h(\mathcal{M}) := \sum_{R_i \in \mathcal{M}} h_i$ and $w(\mathcal{M}) := \sum_{R_i \in \mathcal{M}} w_i$. We start with two preliminary lemmas.

\begin{lemma}\label{lem:rearrange_sliced} Let $q$ and $d$ be two positive integers. If a box $B$ of height $h$ contains only vertical rectangles of width $1$ that have height at least $h/d$, at most $q$ different heights, then there is a packing of all the rectangles in at most $d{(q + 1)}^d$ vertical containers packed in $B$, and the total area of the containers equals the total area of the vertical rectangles in $B$; a symmetrical statement holds for boxes containing only horizontal rectangles of height $1$.
\end{lemma}
\begin{proof}
	W.l.o.g., we only prove the lemma for the case of vertical rectangles. Consider each slice of width $1$ of $B$. In each slice, sort the rectangles by decreasing height, and move them down so that the bottom of each rectangle is touching either the bottom of the box, or the top of another rectangle. Call the \emph{type} of a slice as the set of different heights of rectangles that it contains. It is not difficult to see that there are at most ${(q+1)}^d$ different types of slices. Sort the slices so that all the slices of the same type appear next to each other. It is easy to see that all the rectangles in the slices of a fixed type can be packed in at most $d$ vertical containers, where each container has the same height as the contained rectangles. By repeating this process for all the slices, we obtain a repacking of all the rectangles in at most $d{(q + 1)}^d$ containers.
\end{proof}

\begin{lemma}\label{lem:fractional_to_integral}
	Given a set $\{R_1, R_2, \dots, R_m\}$ of horizontal (resp. vertical) rectangles and a set $\{C_1, C_2, \dots, C_{t}\}$ of horizontal (resp. vertical) containers such that the rectangles can be packed into the containers allowing horizontal (resp. vertical) slicing of height (resp. width) one, then there is a feasible packing of all but at most $t$ rectangles into the same containers.
\end{lemma}
\begin{proof}
	Let us prove the lemma only for the horizontal case (the vertical case is analogous). W.l.o.g. assume that $w(R_1) \ge w(R_2) \ge \dots \ge w(R_m)$ and also $w(C_1) \ge \dots \ge w(C_t)$.
	
	Start assigning the rectangles iteratively to the first container and stop as soon as the total height of assigned rectangles  becomes strictly larger than $h(C_1)$. By discarding the last assigned rectangle, this gives a feasible packing (without slicing) of all the other assigned rectangles in the first container. Then we proceed similarly with the remaining rectangles and following containers.
	
	Now we show that the above procedure outputs a feasible packing of all but at most $t$ rectangles (the discarded ones) into the containers. Due to feasibility of the packing of the sliced rectangles into the containers, we already have $\sum_{i=1}^m{h(R_i)} \le \sum_{i=1}^{t}{h(C_i)}$.
	Note that the non-empty containers (except possibly the last one) are overfilled if we include the discarded rectangles.
	Thus, the above process assigns all the rectangles. 
	
	To finish the proof, we need to show that if $R_j$ is assigned to container $C_k$ by the above procedure, then $w(R_j) \le w(C_k)$.
	Now as containers $C_1, \dots, C_{k-1}$ are overfilled including the so far discarded rectangles, we have that \[\sum_{i=1}^j{h(R_i)} > \sum_{i=1}^{k-1}{h(C_i)}.\]
	
	Now for the sake of contradiction, let us assume $w(R_j) > w(C_k)$.
	Then $w(R_p) > w(C_q)$ for all $p\le j$ and $q \ge k$.
	Thus in every feasible packing, even allowing slicing, rectangles $R_1,\dots, R_j$ must be assigned to containers $C_1, \dots, C_{k-1}$. This contradicts the above inequality.
\end{proof}

Consider the packing obtained by applying Lemma~\ref{lem:structural_boxes}. We will refine this packing to obtain the structural properties claimed in Lemma~\ref{lem:structural_containers}.

\subsection{Horizontal rectangles}
The following lemma allows us to pack horizontal rectangles into a constant number of containers efficiently in terms of area while using negligible extra height.

\begin{lemma}\label{lem:structural_containers_horizontal}
	There is a constant $K_H \leq \left(\frac{1}{\eps \delta_w}\right)^{O(1/(\delta_w \eps))}$ such that, assuming $\mu_h \leq \frac{\eps}{K_H}$, it is possible to pack all the rectangles in $H$ in $K_H \leq \left(\frac{1}{\eps \delta_w}\right)^{O(1/(\delta_w \eps))}$ horizontal containers, so that each container is packed in a box $B \in \mathcal{B}_H \cup \{B_{H, cut}\}$, plus an additional container $B_{H,round}$ of size $W\times O(\eps)OPT'$, and the total area of the containers packed in a box of $\mathcal{B}_H$ is at most $a(H)$.
\end{lemma}
\begin{proof}
	
	Observe that $OPT\cdot W \geq h(H) \delta_w W$. Thus, if $OPT \leq \frac{1}{\eps}$, the statement is immediately proved by defining a container for each rectangle in $H$ (and leaving $B_{H, round}$ empty); since $|H| \leq h(H)$ (being the heights positive integers), this introduces at most $\frac{1}{\eps \delta_w}$ containers. Thus, without loss of generality, we can assume that $OPT > \frac{1}{\eps}$. If $h(H) \leq \left\lceil\eps OPT \right\rceil$, then we can pile all the rectangles in $H$ in $B_{H, round}$, whose height will clearly be at most $h(H) = O(\eps)OPT$. Assume now $h(H) > \left\lceil\eps OPT \right\rceil$.
	
	We use the standard technique of \emph{linear grouping}~\cite{KenyonR00}. Let $j$ be the smallest positive integer such that the set $H_{long}$ consisting of the $j$ horizontal rectangles of maximum width (breaking ties arbitrarily) has height $h(H_{long}) \geq \left\lceil\eps OPT \right\rceil$. Clearly, $h(H_{long}) \leq \left\lceil\eps OPT \right\rceil + \mu_h OPT \leq 3\eps OPT$. We remove the rectangles in $H_{long}$ from the packing. Suppose now that the remaining rectangles are sorted in order of non-increasing width, and that they are sliced in rectangles of unit height. We can form groups $H_1, H_2, \dots, H_t$ of total height exactly $\left\lceil\eps OPT \right\rceil$ (possibly except for the last group, that can have smaller total height). Since $h(H) \leq OPT/\delta_w$, it follows that $t \leq \frac{1}{\eps\delta_w}$. With the convention that $H_0 := H_{long}$, then for each positive integer $i \leq t$ we have that the width of any (possibly sliced) rectangle in $H_t$ is smaller than the width of any rectangle in $H_{t-1}$; round up the widths of each rectangle in $H_i$ to $w_{max}(H_i)$, and let $\overline{H}_i$ be the obtained set of rectangles; let $\overline{H} = \bigcup_{i = 1}^t \overline{H}_i$. By the above observation, for each $i > 0$ it is possible to pack all the rectangles in $\overline{H}_i$ in the space that was occupied in the original packing by the rectangles in $H_{i - 1}$; moreover, $a(\overline{H}) \leq a(H)$.
	
	Consider each box $B \in \mathcal{B}_H \cup \{B_{H,cut}\}$ and the packing of the elements of $\overline{H}$ obtained by the above process. By applying Lemma~\ref{lem:rearrange_sliced} on each box, there is a packing of all the rectangles of $\overline{H}$ in at most $(1/\delta_w){\left(1 + \frac{1}{\eps \delta_w}\right)}^{1/\delta_w} \leq {\left(\frac{1}{\eps\delta_w}\right)}^{3/\delta_w}$ horizontal containers for each box, such that the total area of these containers is at most $a(\overline{H}) \le a(H)$.
	
	By putting back the slices of the original width, we obtain a packing of all the slices of the rectangles in $H_1, H_2, \dots, H_t$. By Lemma~\ref{lem:fractional_to_integral}, there exists a packing of all the rectangles in $H$, except for a set of at most $K_H := K_B {\left(\frac{1}{\eps\delta_w}\right)}^{3/\delta_w}$ horizontal rectangles. Provided that $\mu_h \leq \frac{\eps}{K_H}$, those remaining rectangles can be piled in $B_{H,round}$, together with rectangles in $H_{long}$, by defining its height as $4\eps OPT$.
	
\end{proof}

\subsection{Vertical and tall rectangles}
The main goal of this section is to prove the following lemma:
\begin{lemma}\label{lem:structural_containers_vertical}
	There is a constant $K_V \leq \left(\frac{1}{\eps \delta_w}\right)^{O(1/(\delta_w \eps))}$ such that, assuming $\mu_w \leq \frac{\gamma}{3K_V}$, it is possible to pack all the rectangles in $T \cup V$ in at most $K_V$ vertical containers, so that each container is packed completely either:
	\begin{itemize}
		\item in one of the boxes in $\mathcal{B}_V$; 
		\item in the original position of a nicely cut rectangle from Lemma \ref{lem:boxpartition} and containing only the corresponding nicely cut rectangle;
		\item in a box $B_{disc}$ of size $(1 - \gamma)W \times (1+\eps-2\alpha)OPT'$;
		\item in one of two boxes $B_{V, cut}$ and $B_{V, round}$, each of size $\frac{\gamma}{3} W \times \alpha OPT$, which are in fact containers.
	\end{itemize}
	Moreover, the area of the vertical containers packed in $B_{OPT'}$ is at most $a(T \cup V)$.
\end{lemma}
Consider a specific vertical box $\overline{B}$ of size $\overline{w}\times \overline{h}$; as described in Section~\ref{sec:repack}, the rectangles are repacked so that each rectangle in $\overline{T}$ touches either the top or the bottom edge of $\overline{B}$, and then the set $\overline{P}$ of pseudo rectangles plus the (up to four) corner sub-boxes $\overline{B}_{corn}$ are defined, each one of them containing only slices of rectangles in $V$. Let $\overline{B}_{rem}:=\overline{B} - \overline{T}_{cut}$. \walr{Rearranged this description to match section 3}We now get a rearrangement of this packing applying the following lemma from \cite{NW16}:

\begin{lemma}[follows from the proof of Lemma~3.6 and Section 4 in \cite{NW16}] \label{lem_rearrangeTUP} There is packing of $\overline{T} \cup \overline{P} \cup \overline{B}_{corn}$ into at most $K_R := 2\frac{1+\eps}{\gamma} \cdot 6^{(1+\eps)/\gamma}\wal{+4}$ sub-boxes inside $\overline{B}_{rem}$, such that:
	\begin{itemize}
		\item each sub-box contains only tall rectangles or only pseudo rectangles, that are all of the same height as the sub-box;
		\item \wal{each sub-box is completely occupied by the contained pseudo/tall rectangles, and the $y$-coordinate of such rectangles is the same as before the rearrangement};
		\item the corner sub-boxes in $\overline{B}_{corn}$ and the rectangle slices inside them are packed in the same position as before the rearrangement.
	\end{itemize}
\end{lemma}
\arir{We have been earlier citing lemmas from the submitted version of the NW16 paper, we should cite the lemmas from the published version. Waldo, please change Lemma 10 to 3.6 and Lemma 4 to Lemma 3.2 in  \cite{NW16} etc. \wal{Checked it, I think it is fine now}}

\begin{proof}
We give a brief outline of the proof, the details can be found in Section 4 in \cite{NW16}.\\ First, similar to Lemma 4.6 in  \cite{NW16}, we can combine rectangles in $\overline{B}_{corn}$ with $\overline{T}_{cut}$ to form new unmovable items, \wal{which we denote by $\overline{T}'_{cut}$}. This way we can assume that the boundary of each item in $\overline{T}'_{cut}$ intersects a corner of $\overline{B}$. 
\arir{We need to explain this more. We actually move rectangles in $\overline{B}_{corn}$ to form sub-boxes and may be we need to add the number of such boxes in $K_R$ too. It is a bit tricky as the Lemma in \cite{NW16} was stated before formation of sub-boxes. \wal{Changed it so that sub-boxes in $\overline{B}_{corn}$ are included in $K_R$ (each of them contains one pseudo-rectangle and in this rearrangement we don't move anything inside them)}}
\\
\wal{Recall that items in $\overline{T} \cup \overline{P}$ touch either the top or the bottom boundary of $\overline{B}$.}\\
\ari{Now the following result  can be proven by induction, as shown in Section 4.1 of \cite{NW16}:}\\
Given a packing into a box $\overline{B}$ such that:
\begin{itemize}
\item each item touches the top or the bottom boundary of $\overline{B}$;
\item the height of each item equals one out of at most $\Gamma$ many values;
\item the heights of the items touching the bottom boundary have at most $k$ distinct values;
\item the items touching the four corners are called unmovable items, all other items are movable items;
\end{itemize}
then there exists another packing that does not change the positions of the unmovable items and allows a {\em nice} partition into $6^k \cdot \Gamma$ sub-boxes for the movable items.\\
In this {\em nice} partition, the \emph{sub-boxes} are induced by maximal consecutive subsets of \wal{movable items} of the same height touching the top (resp., bottom) side of $\overline{B}$.
In our case, from Lemma 4, we get $k=\frac{1+\eps}{\gamma}$ and $\Gamma=\frac{1+\eps}{\gamma}$. Now each sub-box can be divided into two sub-boxes by rearranging tall/pseudo rectangles inside: one sub-box contains only tall rectangles while the other one contains only pseudo-rectangles. \wal{By considering also the corner sub-boxes $\overline{B}_{corn}$} we get the desired value of $K_R$.
Furthermore, each sub-box contains only tall rectangles or only pseudo rectangles that are all of the same height as the sub-box \ari{(notice that this holds for corner sub-boxes in $\overline{B}_{corn}$ as well since each one of them contains only pseudo-rectangles of the same height)}.
On the other hand, \wal{in this procedure every rectangle is moved only horizontally, implying that the $y$-coordinate of each pseudo/tall rectangle in $\overline{B}$ remains unchanged after the rearrangement.}
\end{proof}

Consider the packing obtained by the above lemma; partition all the free space in $\overline{B}_{rem}$ which is not occupied by the above defined boxes into at most $2K_R + 1$ empty sub-boxes by considering the maximal rectangular regions that are not intersected by the vertical lines passing through the edges of the sub-boxes. By Lemma~\ref{lem:repack}, the fraction of the rectangles contained in slices of $\overline{D}$ of total width at least $\gamma \overline{w}$ can be repacked inside the empty sub-boxes.

Among the at most $3K_R + 1$ sub-boxes that we defined, some only contain tall rectangles, while the others contain pseudo rectangles. The ones that only contain tall rectangles are already containers and box $B_{V,cut}$ defined in the proof of Lemma~\ref{lem:structural_boxes} is already a vertical container as well. For each sub-box $B'$ that contains pseudo rectangles, we now consider the sliced vertical rectangles that are packed in it. By Lemma~\ref{lem:rearrange_sliced}, there is a packing of all the (sliced) rectangles in $B'$ into at most $\frac{1}{\delta_h}{(1 + 1/\gamma)}^{1/\delta_h}$ containers, and their total area is equal to the total area of the slices of the rectangles they contain. 
There are also at most $4K_B$ containers to pack the tall rectangles that are nicely cut; each of them is packed in his original position in a vertical container of exactly the same size. In total we defined at most $\kappa := (3 K_R + 1) \frac{1}{\delta_h}{(1 + 1/\gamma)}^{1/\delta_h}+4K_B + 1$ containers (where the additional term $1$ is added to take $B_{V, cut}$ into account). We remark that all the tall rectangles are integrally packed, while vertical rectangles are sliced and packed into containers with only slices of vertical rectangles. 
The total area of all the vertical containers packed in $B_{OPT'}$ is at most the sum of the total area of tall items and the total area of the sliced vertical rectangles, i.e., at most $a(T \cup V)$. Finally, by Lemma~\ref{lem:fractional_to_integral}, all but $\kappa$ vertical rectangles can be packed in the containers. With the condition that $\mu_w \leq \frac{\gamma}{3 \kappa}$, these remaining vertical rectangles can be packed in a vertical container $B_{V, round}$ of size $\frac{\gamma W}{3} \times \alpha OPT$. This concludes the proof of Lemma~\ref{lem:structural_containers_vertical} with $K_V := \kappa + 1$.

\subsection{Concluding the proof}

There are at most $K_L := \frac{1}{\delta_h \delta_w}$ many large rectangles. Each such large rectangle is assigned to one container of the same size.

Rectangles in $M$ are packed as described in the proof of  Lemma~\ref{lem:mediumrectanglesrepacking}, using at most $K_M:=\frac{\gamma}{3 \mu_w}+\frac{3 \eps}{\mu_h}$ containers, which are placed in the boxes $B_{M,hor}$ and $B_{M,ver}$.

Horizontal and vertical rectangles are packed as explained in Lemma~\ref{lem:structural_containers_horizontal} and Lemma~\ref{lem:structural_containers_vertical}, respectively. The total number of containers $K_{TOTAL}=K_L+K_M+K_H+K_V$, is clearly $O_\eps(1)$, and each of these containers is either contained in or disjoint from $\mathcal{B}_{OPT'}$. Among them, at most $K_F := K_L + K_H + K_V$ containers lie inside $\mathcal{B}_{OPT'}$. The total area of these $K_F$ containers is at most $a(H)+a(T \cup V)+a(L) \le a(\mathcal{R} \setminus S)$.

By packing the boxes and containers we defined as in Figure~\ref{fig_our_packing}, we obtain a packing in a strip of width $W$ and height $OPT' \cdot (\max\{1 + \alpha, 1 + (1 - 2\alpha)\} + O(\eps))$, which is at most $(4/3 + O(\eps))OPT'$ for $\alpha = 1/3$. This concludes the proof of Lemma~\ref{lem:structural_containers}.

\section{The final algorithm}
\label{sec:algo}

First of all, we find $\mu_h, \delta_h, \mu_w, \delta_w$ as required by Lemma~\ref{lem:mediumrectanglesarea}; this way, we can find the set $S$ of small rectangles. Consider the packing of Lemma~\ref{lem:structural_containers}: all the non-small rectangles are packed into $K_{TOTAL}=O_{\eps}(1)$ containers, and only $K_F$ of them are contained in $B_{OPT'}$. Since their position $(x,y)$ and their size $(w, h)$ are w.l.o.g. contained in $\{0,\ldots,W\}\times \{0,\ldots,nh_{max}\}$, we can enumerate in PPT over all the possible feasible such packings of $k \leq K_{TOTAL}$ containers, and one of those will coincide with the packing defined by Lemma~\ref{lem:structural_containers}.

Containers naturally induce a multiple knapsack problem: for each horizontal container $C_j$ of size $w_{C_j} \times h_{C_j}$, we create a (one-dimensional) knapsack $j$ of size $h_{C_j}$. Furthermore, we define the size $b(i,j)$ of rectangle $R_i$ w.r.t. knapsack $j$ as $h_i$ if $h_i\leq h_{C_j}$ and $w_i\leq w_{C_j}$. Otherwise $b(i,j)=+\infty$ (meaning that $R_i$ does not fit in $C_j$). The construction for vertical containers is symmetric. This multiple knapsack problem can be easily solved optimally (hence packing all the rectangles) in PPT via dynamic programming.

Note that unlike \cite{NW16}, we do not use linear programming to pack horizontal rectangles, which will be crucial when we extend our approach to the case with rotations.

\subsection{Packing the small rectangles}
\label{sec:smallpack}
It remains to pack the small rectangles $S$. We will pack them in the free space left by containers inside $[0,W] \times [0,OPT']$ plus an additional box $B_S$ of small height as the following lemma states. By placing box $B_S$ on top of the remaining packed rectangles, the final height of the solution increases only by $\eps \cdot OPT'$.
\begin{lemma}\label{lem:smallpacking} Assuming $\mu_h \le \frac{1}{31 K_F^2}$, it is possible to pack in polynomial time all the rectangles in $S$ into the area $[0,W] \times [0,OPT']$ not occupied by containers plus an additional box $B_S$ of size $W \times \eps OPT'$. \end{lemma}

\begin{proof} We first extend the sides of the containers inside $[0,W] \times [0,OPT']$ in order to define a grid. This procedure partitions the free space in $[0,W] \times [0,OPT']$ into a constant number of rectangular regions (at most ${(2K_F+1)}^2 \leq 5K_F^2$ many) whose total area is at least $a(S)$ thanks to Lemma~\ref{lem:structural_containers}. Let $\mathcal{B}_{small}$ be the set of such rectangular regions with width at least $\mu_w W$ and height at least $\mu_h OPT$ (notice that the total area of rectangular regions not in $\mathcal{B}_{small}$ is at most $5K_F^2\mu_w \mu_h \cdot W \cdot OPT$). We now use NFDH to pack a subset of $S$ into the regions in $\mathcal{B}_{small}$. Thanks to Lemma~\ref{lem:NFDHarea}, since each region in $\mathcal{B}_{small}$ has size at most $W \times OPT'$ and each item in $S$ has width at most $\mu_w W$ and height at most $\mu_h OPT$, the total area of the unpacked rectangles from $S$ can be bounded above by $5K_F^2\cdot\Big(\mu_w \mu_h W OPT + \mu_h OPT \cdot W + \mu_w W \cdot OPT' \Big) \le 15K_F^2 \mu_h \cdot OPT' \cdot W$. Therefore, thanks to Lemma~\ref{lem:NFDH}, we can pack the latter small rectangles with NFDH in an additional box $B_S$ of width $W$ and height $\mu_h OPT + 30K_F^2 \mu_h OPT' \le \eps\cdot OPT'$ provided that $\mu_h \le \frac{1}{31K_F^2}$. \end{proof}

We next summarize the constraints that arise from the analysis:
\begin{table}[ht]
\centering
{\normalsize
\begin{tabular}{l l}
     $\bullet$ $\mu_w = \frac{\eps \mu_h}{12}$ and $\delta_w = \frac{\eps \delta_h}{12}$ (Lemma~\ref{lem:mediumrectanglesarea}), \qquad\qquad &
     $\bullet$ $\mu_w \leq \gamma \frac{\delta_h}{6K_B(1+\eps)}$ (Lemma~\ref{lem:structural_boxes}),\\

     $\bullet$ $\gamma = \frac{\eps \delta_h}{2}$ (Lemma~\ref{lem:verticalrounding}), &
	 $\bullet$ $\mu_w \leq \frac{\gamma}{3 K_F}$ (Lemma~\ref{lem:structural_containers}),\\

     $\bullet$ $6 \eps^k \le \frac{\gamma}{6}$ (Lemma \ref{lem:mediumrectanglesrepacking}) &
	 $\bullet$ $\mu_h \leq \frac{\eps}{K_F}$ (Lemma \ref{lem:structural_containers}),\\

     $\bullet$ $\mu_h \leq \frac{\eps \delta_w}{K_B}$ (Lemma~\ref{lem:boxpartition}), &
	 $\bullet$ $\mu_h \leq \frac{1}{31 K_F^2}$ (Lemma~\ref{lem:smallpacking}) 
\end{tabular}
}
\end{table}

It is not difficult to see that all the constraints are satisfied by choosing $f(x) = (\eps x)^{C/(\eps x)}$ for a large enough constant $C$ and $k = \left\lceil\log_{\eps}\left(\frac{\gamma}{36}\right)\right\rceil$. Finally we achieve the claimed result.

\begin{theorem} There is a PPT $(\frac{4}{3}+\eps)$-approximation algorithm for strip packing. 
\end{theorem}

\section{Extension to the case with rotations}
\label{sec:rot}
In this section, we briefly explain the changes needed in the above algorithm to handle the case with rotations.

We first observe that, by considering the rotation of rectangles as in the optimum solution, Lemma \ref{lem:structural_containers} still applies (for a proper choice of the parameters, that can be guessed). Therefore we can define a multiple knapsack instance, where knapsack sizes are defined as before. Some extra care is needed to define the size $b(i,j)$ of rectangle $R_i$ into a container $C_j$ of size $w_{C_j}\times h_{C_j}$. Assume $C_j$ is horizontal, the other case being symmetric. If rectangle $R_i$ fits in $C_j$ both rotated and non-rotated, then we set $b(i,j)=\min\{w_i,h_i\}$ (this dominates the size occupied in the knapsack by the optimal rotation of $R_i$). If $R_i$ fits in $C_j$ only non-rotated (resp., rotated), we set $b(i,j)=h_i$ (resp., $b(i,j)=w_i$). Otherwise we set $b(i,j)=+\infty$.

There is a final difficulty that we need to address: we can not say a priori whether a rectangle is small (and therefore should be packed in the final stage). To circumvent this difficulty, we define one extra knapsack $k'$ whose size is the total area in $\mathcal{B}_{OPT'}$ not occupied by the containers. The size $b(i,k')$ of $R_i$ in this knapsack is the area $a(R_i)=w_i\cdot h_i$ of $R_i$ provided that $R_i$ or its rotation by $90^\circ$ is small w.r.t. the current choice of the parameters $(\delta_h,\mu_h,\delta_w,\mu_w)$. Otherwise $b(i,k')=+\infty$.

By construction, the above multiple knapsack instance admits a feasible solution that packs all the rectangles. This immediately implies a packing of all the rectangles, excluding the (small) ones in the extra knapsack. Those rectangles can be packed using NFDH as in the proof of Lemma \ref{lem:smallpacking} (here however we must choose a rotation such that the considered rectangle is small). Altogether we achieve: 
\begin{theorem} There is a PPT $(\frac{4}{3}+\eps)$-approximation algorithm for strip packing with rotations. 
\end{theorem}

\section{Conclusions}
\label{sec:conc}

In this paper we obtained a PPT $(4/3+\eps)$-approximation for strip packing (with and without rotations). Our approach refines and, in some sense, pushes to its limit the basic approach in the previous work by Nadiradze and Wiese \cite{NW16}. 
Indeed, the rearrangement of rectangles inside a box crucially exploits the fact that there are at most 2 tall rectangles packed on top of each other in the optimal packing, hence requiring $\alpha \geq 1/3$. 
\ari{It will be interesting to settle the complexity of the problem by providing matching (PPT) approximation ratio and hardness of approximation.}

\walr{this last paragraph needs to be rewritten \ari{made shorter.}}

\bibliographystyle{plain}
\bibliography{bibliography}{}

\begin{thebibliography}{10}

\bibitem{AKPP17}
Anna Adamaszek, Tomasz Kociumaka, Marcin Pilipczuk, and Micha\l Pilipczuk.
\newblock Hardness of approximation for strip packing.
\newblock {\em ACM Trans. Comput. Theory}, 9(3):14:1--14:7, 2017.

\bibitem{AW13}
Anna Adamaszek and Andreas Wiese.
\newblock Approximation schemes for maximum weight independent set of
  rectangles.
\newblock In {\em FOCS}, pages 400--409, 2013.

\bibitem{AdamaszekW15}
Anna Adamaszek and Andreas Wiese.
\newblock A quasi-{PTAS} for the two-dimensional geometric knapsack problem.
\newblock In {\em SODA}, pages 1491--1505, 2015.

\bibitem{AGLW14}
Aris Anagnostopoulos, Fabrizio Grandoni, Stefano Leonardi, and Andreas Wiese.
\newblock A mazing 2+$\varepsilon$ approximation for unsplittable flow on a
  path.
\newblock In {\em SODA}, pages 26--41, 2014.

\bibitem{baker198154}
Brenda~S. Baker, Donna~J. Brown, and Howard~P. Katseff.
\newblock A 5/4 algorithm for two-dimensional packing.
\newblock {\em J. Algorithm}, 2(4):348--368, 1981.

\bibitem{BakerCR80}
Brenda~S. Baker, Edward G.~Coffman Jr., and Ronald~L. Rivest.
\newblock Orthogonal packings in two dimensions.
\newblock {\em {SIAM} J. Comput.}, 9(4):846--855, 1980.

\bibitem{BCES06}
Nikhil Bansal, Amit Chakrabarti, Amir Epstein, and Baruch Schieber.
\newblock A quasi-{PTAS} for unsplittable flow on line graphs.
\newblock In {\em STOC}, pages 721--729, 2006.

\bibitem{BansalK14}
Nikhil Bansal and Arindam Khan.
\newblock Improved approximation algorithm for two-dimensional bin packing.
\newblock In {\em SODA}, pages 13--25, 2014.

\bibitem{BSW14}
Paul~S. Bonsma, Jens Schulz, and Andreas Wiese.
\newblock A constant-factor approximation algorithm for unsplittable flow on
  paths.
\newblock {\em {SIAM} J. Comput.}, 43(2):767--799, 2014.

\bibitem{CC09}
Parinya Chalermsook and Julia Chuzhoy.
\newblock Maximum independent set of rectangles.
\newblock In {\em SODA}, pages 892--901, 2009.

\bibitem{CH12}
Timothy~M. Chan and Sariel Har{-}Peled.
\newblock Approximation algorithms for maximum independent set of pseudo-disks.
\newblock {\em Discrete {\&} Computational Geometry}, 48(2):373--392, 2012.

\bibitem{ChristensenKPT17}
Henrik~I. Christensen, Arindam Khan, Sebastian Pokutta, and Prasad Tetali.
\newblock Approximation and online algorithms for multidimensional bin packing:
  {A} survey.
\newblock {\em Computer Science Review}, 24:63--79, 2017.

\bibitem{coffman1976computer}
Edward~G. Coffman~Jr. and John~L. Bruno.
\newblock {\em Computer and job-shop scheduling theory}.
\newblock John Wiley \& Sons, 1976.

\bibitem{coffman2013bin}
Edward~G. Coffman~Jr., J{\'a}nos Csirik, G{\'a}bor Galambos, Silvano Martello,
  and Daniele Vigo.
\newblock Bin packing approximation algorithms: survey and classification.
\newblock In {\em Handbook of Combinatorial Optimization}, pages 455--531.
  Springer, 2013.

\bibitem{CGJT80}
Edward~G. Coffman~Jr., Michael~R. Garey, David~S. Johnson, and Robert~E.
  Tarjan.
\newblock Performance bounds for level-oriented two-dimensional packing
  algorithms.
\newblock {\em SIAM J. Comput.}, 9(4):808--826, 1980.

\bibitem{epstein2006side}
Leah Epstein and Rob van Stee.
\newblock This side up!
\newblock {\em ACM Transactions on Algorithms (TALG)}, 2(2):228--243, 2006.

\bibitem{GGHIKW17}
Waldo G{\'{a}}lvez, Fabrizio Grandoni, Sandy Heydrich, Salvatore Ingala,
  Arindam Khan, and Andreas Wiese.
\newblock Approximating geometric knapsack via {L}-packings.
\newblock In {\em FOCS}, pages 260--271, 2017.

\bibitem{GalvezGIK16}
Waldo G{\'{a}}lvez, Fabrizio Grandoni, Salvatore Ingala, and Arindam Khan.
\newblock Improved pseudo-polynomial-time approximation for strip packing.
\newblock In {\em FSTTCS}, pages 9:1--9:14, 2016.

\bibitem{garey1978strong}
Michael~R. Garey and David~S. Johnson.
\newblock ``{S}trong'' {NP}-completeness results: Motivation, examples, and
  implications.
\newblock {\em JACM}, 25(3):499--508, 1978.

\bibitem{golan1981performance}
Igal Golan.
\newblock Performance bounds for orthogonal oriented two-dimensional packing
  algorithms.
\newblock {\em SIAM J. Comput.}, 10(3):571--582, 1981.

\bibitem{harren20145}
Rolf Harren, Klaus Jansen, Lars Pr{\"a}del, and Rob van Stee.
\newblock A (5/3+ $\varepsilon$)-approximation for strip packing.
\newblock {\em Computational Geometry}, 47(2):248--267, 2014.

\bibitem{harren2009improved}
Rolf Harren and Rob van Stee.
\newblock Improved absolute approximation ratios for two-dimensional packing
  problems.
\newblock In {\em APPROX-RANDOM}, pages 177--189. Springer, 2009.

\bibitem{HJRS17}
S{\"{o}}ren Henning, Klaus Jansen, Malin Rau, and Lars Schmarje.
\newblock Complexity and inapproximability results for parallel task scheduling
  and strip packing.
\newblock {\em CoRR}, abs/1705.04587, 2017.

\bibitem{JansenP14}
Klaus Jansen and Lars Pr{\"{a}}del.
\newblock A new asymptotic approximation algorithm for 3-dimensional strip
  packing.
\newblock In {\em SOFSEM}, pages 327--338, 2014.

\bibitem{Raunew}
Klaus Jansen and Malin Rau.
\newblock Closing the gap for pseudo-polynomial strip packing.
\newblock {\em CoRR}, abs/1712.04922, 2017.

\bibitem{JansenR17}
Klaus Jansen and Malin Rau.
\newblock Improved approximation for two dimensional strip packing with
  polynomial bounded width.
\newblock In {\em WALCOM}, pages 409--420, 2017.

\bibitem{jansen2009rectangle}
Klaus Jansen and Roberto Solis-Oba.
\newblock Rectangle packing with one-dimensional resource augmentation.
\newblock {\em Discrete Optimization}, 6(3):310--323, 2009.

\bibitem{JT10}
Klaus Jansen and Ralf Th\"ole.
\newblock Approximation algorithms for scheduling parallel jobs.
\newblock {\em SIAM J. Comput.}, 39(8):3571--3615, 2010.

\bibitem{jansen2005strip}
Klaus Jansen and Rob van Stee.
\newblock On strip packing with rotations.
\newblock In {\em STOC}, pages 755--761, 2005.

\bibitem{JansenZ07}
Klaus Jansen and Guochuan Zhang.
\newblock Maximizing the total profit of rectangles packed into a rectangle.
\newblock {\em Algorithmica}, 47(3):323--342, 2007.

\bibitem{karbasioun2013power}
Mohammad~M. Karbasioun, Gennady Shaikhet, Evangelos Kranakis, and Ioannis
  Lambadaris.
\newblock Power strip packing of malleable demands in smart grid.
\newblock In {\em IEEE International Conference on Communications}, pages
  4261--4265, 2013.

\bibitem{KenyonR00}
Claire Kenyon and Eric R{\'e}mila.
\newblock A near-optimal solution to a two-dimensional cutting stock problem.
\newblock {\em Math. Oper. Res.}, 25(4):645--656, 2000.

\bibitem{KhanThesis}
Arindam Khan.
\newblock {\em Approximation Algorithms for Multidimensional Bin Packing}.
\newblock PhD thesis, Georgia Institute of Technology, USA, 2015.

\bibitem{miyazawa2004packing}
Flavio~Keidi Miyazawa and Yoshiko Wakabayashi.
\newblock Packing problems with orthogonal rotations.
\newblock In {\em LATIN}, pages 359--368. Springer, 2004.

\bibitem{NW16}
Giorgi Nadiradze and Andreas Wiese.
\newblock On approximating strip packing better than 3/2.
\newblock In {\em SODA}, pages 1491--1510, 2016.

\bibitem{ranjan2015offline}
Anshu Ranjan, Pramod Khargonekar, and Sartaj Sahni.
\newblock Offline first fit scheduling in smart grids.
\newblock In {\em IEEE SCC}, pages 758--763, 2015.

\bibitem{schiermeyer1994reverse}
Ingo Schiermeyer.
\newblock Reverse-fit: A 2-optimal algorithm for packing rectangles.
\newblock In {\em ESA}, pages 290--299. Springer, 1994.

\bibitem{sleator19802}
Daniel Sleator.
\newblock A 2.5 times optimal algorithm for packing in two dimensions.
\newblock {\em Information Processing Letters}, 10(1):37--40, 1980.

\bibitem{steinberg1997strip}
A.~Steinberg.
\newblock A strip-packing algorithm with absolute performance bound 2.
\newblock {\em SIAM J. Comput}, 26(2):401--409, 1997.

\bibitem{tang2013smoothing}
Shaojie Tang, Qiuyuan Huang, Xiang-Yang Li, and Dapeng Wu.
\newblock Smoothing the energy consumption: Peak demand reduction in smart
  grid.
\newblock In {\em INFOCOM}, pages 1133--1141. IEEE, 2013.

\end{thebibliography}

\end{document}